 \documentclass[11pt]{article}
 \usepackage{bbm}
 \usepackage{amsthm}
\usepackage{amsmath,amssymb}

\usepackage{amsfonts}            
\usepackage{mathrsfs}            
\usepackage[numbers,sort&compress]{natbib}
\usepackage{booktabs} 
\usepackage{bbm}
\usepackage{bm}
\usepackage{color,xcolor}
\usepackage{algorithm}
\usepackage{algpseudocode}
\usepackage{url}
\ifx\pdfoutput\undefined 
\usepackage{graphicx}
\DeclareGraphicsExtensions{.eps,.ps} \else
\usepackage{graphicx}
\DeclareGraphicsExtensions{.pdf} \fi
\graphicspath{{figs/}}%
\usepackage{float}
\usepackage{subfigure}
\usepackage{diagbox}

\def\MatrixFont{\bf}
\def\VectorFont{\bf}

\newcommand{\mH}{{\MatrixFont H}}
\newcommand{\mI}{{\MatrixFont I}}

\newcommand{\mK}{{\MatrixFont K}}

\newcommand{\mP}{{\MatrixFont P}}
\newcommand{\mQ}{{\MatrixFont Q}}
\newcommand{\mR}{{\MatrixFont R}}

\newcommand{\mW}{{\MatrixFont W}}
\newcommand{\mX}{{\MatrixFont X}}
\newcommand{\mY}{{\MatrixFont Y}}
\newcommand{\mZ}{{\MatrixFont Z}}

\newcommand{\vm}{{\VectorFont m}}

\newcommand{\vp}{{\VectorFont p}}

\newcommand{\vv}{{\VectorFont v}}
\newcommand{\vw}{{\VectorFont w}}
\newcommand{\vx}{{\VectorFont x}}

\newcommand{\vz}{{\VectorFont z}}

\newtheorem{proposition}{Proposition}

\newtheorem{lemma}{Lemma}
\newtheorem{remark}{Remark}

\usepackage{amssymb}
\usepackage{epsfig}
\usepackage{times}
\title{The Effect of Population Size for Pathogen Transmission on Prediction of COVID-19 Pandemic Spread}

\begin{document}
\author{Xuqi Zhang$^{\text{1}}$, Haiqi Liu$^{\text{1}}$, Hanning Tang$^{\text{1}}$, Mei Zhang$^{\text{1}}$, Xuedong Yuan$^{\text{2}}$, Xiaojing Shen$^{\text{1}*}$\\
\textsuperscript{1}School of Mathematics, Sichuan University, Chengdu, Sichuan, China\\
\textsuperscript{2}School of Computer Science, Sichuan University, Chengdu, Sichuan, China}
\renewcommand{\thefootnote}{\fnsymbol{footnote}}
\footnotetext[1]{Corresponding author (shenxj@scu.edu.cn).}
\footnotetext[2]{This work was supported in part by the NSFC No. 61673282.}
\date{}
 \maketitle
\begin{abstract}
Extreme public health interventions play a critical role in mitigating the local and global prevalence
and pandemic potential of COVID-19. Here, we use population size for pathogen transmission to measure the intensity of public health interventions, which is a key characteristic variable for nowcasting and forecasting of the epidemic. By formulating a hidden Markov dynamic system and using nonlinear filtering theory, we have developed a stochastic epidemic dynamic model under public health interventions. The model parameters and states are estimated in time from internationally available public data by combining an unscented filter and an interacting multiple model filter. Moreover, we consider the computability of the population size and provide its selection criterion. We estimate the mean of the basic reproductive number of China and the rest of the globe except China (GEC) to be 2.46 (95\% CI: 2.41-2.51) and 3.64 (95\% CI: (3.55-3.72), respectively. We infer that the number of latent infections of GEC is about $7.47\times10^5$ (95\% CI: $7.32\times10^5$-$7.62\times10^5$) as of April 2, 2020. We predict that the peak of infections in hospitals of GEC may reach $3.00\times10^6$ on the present trajectory, i.e., if the population size for pathogen transmission and epidemic parameters remains unchanged. If the control intensity is strengthened, e.g., 50\% reduction or 75\% reduction of the population size for pathogen transmission, the peak would decline to $1.84\times10^6$, $1.27\times10^6$, respectively.
\end{abstract}

\noindent{\bf keywords:} 
Coronavirus disease 2019 (COVID-19), stochastic SEIR model, nowcasting, forecasting.

\section{Introduction}
In December 2019, a new type of coronavirus pneumonia named COVID-19 spread very rapidly in the city of Wuhan, China \cite{Ncov2019}. It swept the globe in late January to April, 2020. This outbreak of COVID-19 is now present in at least 100 countries and the virus has infected over $9.64\times10^5$ people with more than $4.95\times10^4$ deaths as of April 2, 2020. This new coronavirus has been considered to be a true global health emergency. The World Health Organization (WHO) announced the COVID-19 to be a global pandemic on March 11, 2020 local time in Geneva. Faced with the severe epidemic trend of COVID-19, many countries have declared a state of emergency and have instituted drastic actions and interventions to stop the spread of the virus. For example, China implemented unprecedented intervention strategies on January 23, 2020  to prevent further outbreak of COVID-19. These policies have included large-scale quarantine, strict controls on travel and extensive monitoring of suspected cases which have made a huge difference in controlling the spread of the virus. Nevertheless, the outbreak is still continuing and having a devastating impact  in Europe, USA and many other countries.  Many scientific research groups in the world including those in  China, Europe, and the United States are working independently  or together by sharing their results and  experiences, including analyzing the mechanism of disease transmission \cite{wu2020new,zhou2020pneumonia,zhong} and identifying optimal control measures \cite{SEIR_cite3,JTD36385,Layneeabb1469} for countries where the outbreak is in its early stages. In addition, the prevention and control experience of the 2003 SARS epidemic \cite{ksiazek2003novel,hu2017discovery,chowell2003sars}, the worldwide 2009 H1N1 influenza pandemic \cite{fraser2009pandemic}, and the 2012 MERS-CoV \cite{haagmans2014middle}  are also being used for COVID-19. 

Mathematical modeling and prediction tools are extremely important tools for decision-making by policy makers for epidemic control. The classical SIR, SEIR and GLEAM models are extensively used \cite{SIR,sciencepop,SEIR,SEIR_cite1,SEIR_cite2,chinazzi2020effect} for comprehending the mechanism of disease transmission, its spread and forecasting. These approaches have become remarkably successful in free infection and transmission patterns and predicting the temporal evolution of ongoing epidemics. However, under strict control measures,  the abundance of different, often
mutually incompatible, forecasting results, suggest that we still lack a fundamental understanding of the key factors of infection and transmission dynamics. One of the them is the population size under consideration for pathogen transmission. Here, we assume that there exists a closed subsystem in a city or country when the uninfected people are in home quarantine or some cities are locked down. The population size of the subsystem is considered as the population size for pathogen transmission, which is usually much smaller than the total population of a city or county when deploying intervention strategies. The total population is divided into the population for pathogen transmission and the isolated uninfected population (see Fig. \ref{population_divide}).
\begin{figure}[H] 
	\centering  
	\vspace{0.35cm} 
	\subfigtopskip=2pt 
	\includegraphics[scale=0.15]{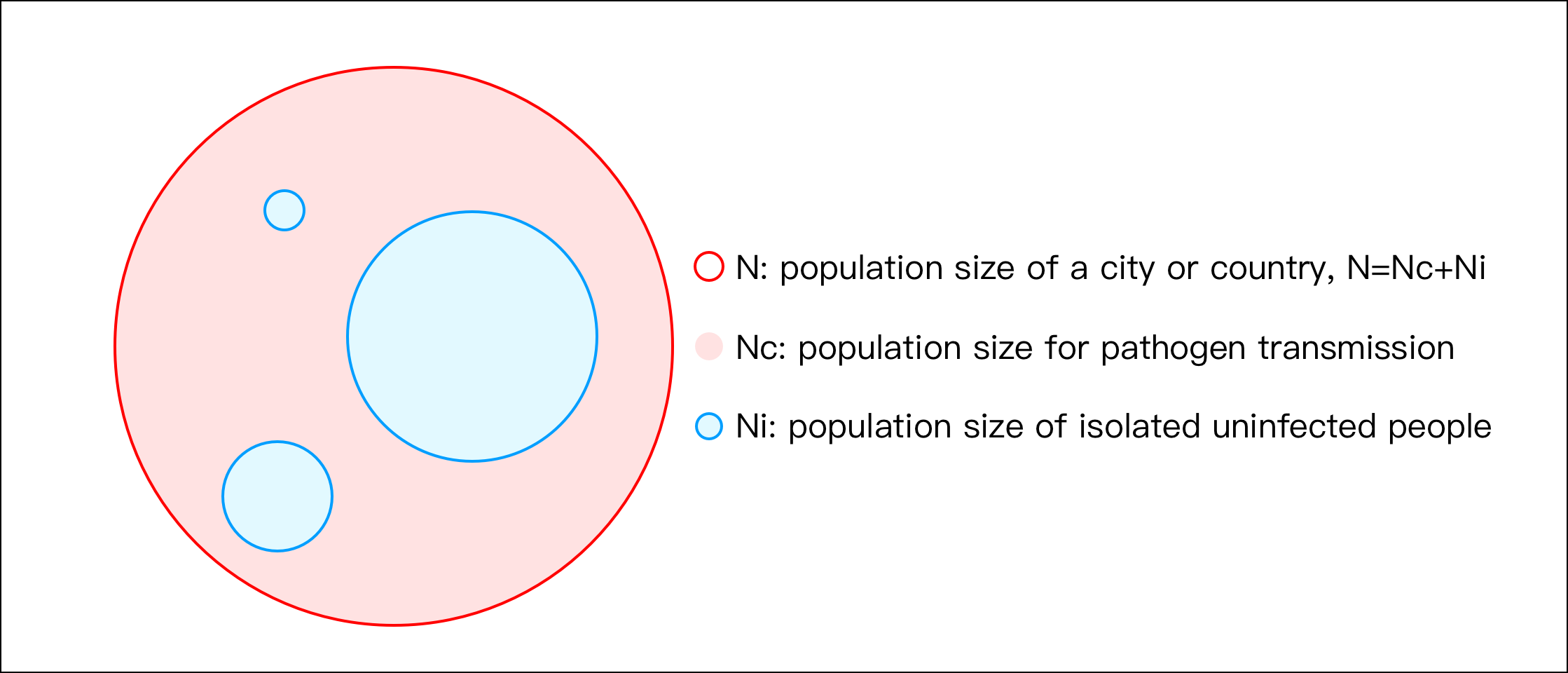}
	\caption{The total population and its division.}
	\label{population_divide}
\end{figure}

In this paper, we use the population size $N_c$ for pathogen transmission to measure the intensity of public health interventions, discuss the identifiability of the population size $N_c$ and propose an intuitive and efficient criterion to select the unknown population size. Furthermore, we concentrate on three specific aims compatible with  WHO's objectives \cite{Layneeabb1469}: (\romannumeral1) Modeling: develop mathematical differential equation dynamic models that account for the random variations in pathogen, society, and public health intervention variables etc.; (\romannumeral2) Nowcasting: estimate the states and parameters as a function of time through nonlinear filtering with publicly available international data of COVID-19. Moreover, we analyze the key factor, the number of latent infections, which determines the size of the newly confirmed infections in the next seven to fourteen days in terms of the incubation period \cite{zhong}, \cite{latent}; (\romannumeral3) Forecasting: predict local and global spread trends (e.g. scales, peaks and confidence intervals) of the infections under different control intensities, i.e., different population sizes for pathogen transmission. Our main contributions are summarized as follows:
\begin{itemize}
	\item We have developed a stochastic epidemic dynamic model under public health
	interventions. By formulating a hidden Markov dynamic system, the model parameters and states are estimated in time from publicly available international data by
	combining an unscented filter and an interacting multiple model filter. 
	
	\item Moreover, we consider the computability of the population size and provide its selection criterion. We use it to
	measure the intensity of public health interventions, which is a key characteristic variable for estimation and prediction of the epidemic. 
	
	\item We provide comprehensive nowcasting and forecasting results for the local and global infections. We estimate the mean of the basic reproductive number of China and GEC to be 2.46 (95\% CI: 2.41-2.51) and
	3.64 (95\% CI: (3.55-3.72), respectively. 
\end{itemize}

This paper is organized as follows: in Section 2, stochastic SEIR models under public health interventions are proposed. And then, the nowcasting and forecasting methods for this problem are introduced in Section 3. Data analysis of COVID-19 is given in Section 4, and the discussion of our results and conclusion are given in Section 5.
\section{Mathematical Epidemic Modeling}
In this section, we introduce a stochastic susceptible-exposed-infectious-recovered (SEIR) model for the free propagation stage of the epidemic, where random variables and processes are included in the epidemiological process model to characterize the uncertainty associated with the propagation of the epidemic. Moreover, we develop a stochastic SEIR model under public health interventions. 
\subsection{The Stochastic SEIR Model}
Based on the classical SEIR model proposed by Aron and Schwartz \cite{SEIR} that is deterministic in nature. In order to account for system uncertainties, we modify the SEIR model by introducing random noises and random parameters in the SEIR model as follows:
\begin{align}
\label{RSEIR1}\frac{\mathrm{d} S(t)}{\mathrm{d} t}&=-\alpha(t) \cdot \frac{S(t) \cdot I(t)}{N}+w_{S}(t),\\
\label{RSEIR2}\frac{\mathrm{d} E(t)}{\mathrm{d} t}&=\alpha(t) \cdot \frac{S(t) \cdot I(t)}{N}-\beta(t)\cdot E(t)+w_{E}(t),\\
\frac{\mathrm{d} I(t)}{\mathrm{d} t}&=\beta(t)\cdot E(t) -(\gamma^1(t)+\gamma^2(t)) \cdot I(t)+w_{I}(t),\\
\frac{\mathrm{d} R(t)}{\mathrm{d} t}&=\gamma^1(t) \cdot I(t)+w_{R}(t),\\
\label{RSEIR5}	\frac{\mathrm{d} D(t)}{\mathrm{d} t}&=\gamma^2(t) \cdot I(t)+w_{D}(t),\\
\label{parameter1}\frac{\mathrm{d}\vp(t)}{\mathrm{d}t}&=\vw_{\vp}(t),
\end{align}
where $S(t)$, $E(t)$, $I(t)$, $R(t)$ and $D(t)$  are the number of the susceptible, exposed, infectious, recovered and disease-caused death cases at time $t$, respectively. We denote the state vector as 
\begin{align*}
\vx(t)=\left(S(t), E(t), I(t), R(t), D(t),\vp(t)^T\right)^T,
\end{align*} 
where the parameter  vector $\vp(t)=\left(\alpha(t), \beta(t), \gamma^1(t), \gamma^2(t)\right)^T$ includes the mean contact rate, transfer rate from exposed to infective, recovery rate and disease-caused mortality. Since these model parameters may experience random changes in the spreading process due to different climate changes, population densities, public health interventions, and medical cares in different regions, we model the uncertainty associated with the parameters as Brownian motions, which could be explained by a random walk with a small noise due to the fact that the parameters represent average random characteristics of a large number of members of the total population \cite{brown}, and $\vw_{\vp}(t)=\left(w_{\alpha}(t), w_{\beta}(t),w_{\gamma^1}(t),w_{\gamma^2}(t)\right)^T$ is a white noise process with  spectral density $\mQ_{\vp}(t)$. Moreover, we define the basic reproductive number $R_0= \alpha/(\gamma^1+\gamma^2)$, which means that the mean number of infections caused by an infected individual in a susceptible population \cite{R0}. Besides, the birth rate, natural death rate and other uncertainties are considered as a zero mean white noise process $\vw(t)=\left(w_{S}(t), w_{E}(t), w_{I}(t),w_{R}(t),w_{D}(t)\right)^T$ with spectral density $\mQ_{\vw}(t)$. In a short epidemic period, the total population size remains constant, and
\begin{align}
\label{cons01}S(t) + E(t) + I(t) +R(t) +D(t)= N.
\end{align}
The stochastic model (\ref{RSEIR1})-(\ref{parameter1}) is suitable for the free propagation of epidemic without control measures. It also characterizes the uncertainties of parameters and states in the propagation process.
\subsection{The Stochastic SEIR Model Under Public Health Interventions}
When there are public health interventions and the infections are isolated, we extend the stochastic SEIR model as follows:
\begin{align}
\label{stochastic SEIRC1}\frac{\mathrm{d} S_c(t)}{\mathrm{d} t}&=-\alpha_c(t) \cdot \frac{S_c(t) \cdot E_c(t)}{N_c}+w_{S_c}(t),\\
\label{stochastic SEIRC2}\frac{\mathrm{d} E_c(t)}{\mathrm{d} t}&=\alpha_c(t) \cdot \frac{S_c(t) \cdot E_c(t)}{N_c}-\beta_c(t)\cdot E_c(t) +w_{E_c}(t),\\
\label{stochastic SEIRC3}\frac{\mathrm{d} I_c(t)}{\mathrm{d} t}&=\beta_c(t)\cdot E_c(t)-(\gamma_c^1(t)+\gamma_c^2(t)) \cdot I_c(t)+w_{I_c}(t),\\
\label{stochastic SEIRC4}\frac{\mathrm{d} R_c(t)}{\mathrm{d} t}&=\gamma_c^1(t) \cdot I_c(t)+w_{R_c}(t),\\
\label{stochastic SEIRC5}\frac{\mathrm{d} D_c(t)}{\mathrm{d} t}&=\gamma_c^2(t) \cdot I_c(t)+w_{D_c}(t),\\
\label{stochastic SEIRC6}\frac{\mathrm{d}\vp_c(t)}{\mathrm{d}t}&=\vw_{\vp_c}(t),
\end{align}
where $S_c$ and $N_c$ are the susceptible and the population size under public health interventions, respectively. $E_c(t)$ and $I_c(t)$ are redefined as the number of latent infections, and confirmed infections at time $t$, respectively. The reason is that the confirmed infections can be observed each day and it is convenient to estimate the model parameters and forecast the number of confirmed and latent infections by nonlinear filtering methods. 
Moreover, the infectious $I(t)$ in (\ref{RSEIR1}) are replaced by $E_c(t)$ in (\ref{stochastic SEIRC1}), which means that the susceptible are not infected by isolated infections. The state vector and the parameters $\vx_c(t), \vp_c(t), \vw_{\vp_c}(t)$ and $\vw_c(t)=\left(w_{S_c}(t), w_{E_c}(t), w_{I_c}(t),w_{R_c}(t),w_{D_c}(t)\right)^T$  are similarly defined as that in the stochastic model (\ref{RSEIR1})-(\ref{parameter1}). 

It is worth noting that under public health interventions, we assume that there exists a closed subsystem in a city
or country when the uninfected people are in home quarantine or some cities are locked down. The
population size of the subsystem is considered as the population size for pathogen transmission $N_c$, which is usually much smaller than the urban or national population when deploying intervention strategies. The total population is divided into the population for pathogen transmission and the isolated uninfected population (see Fig. \ref{population_divide}). In fact, the population size $N_c$ plays a key role in nowcasting and forecasting epidemic trends and measures of control intensity. This paper aims to analyze the data of COVID-19 epidemic when the population size $N_c$ is unknown. The subscript $c$ of $N_c, S_c , E_c, I_c, R_c, D_c, \vp_{c}, \vx_c, \vw_c$ and $\vw_{\vp_c}$ is omitted for simplicity when there is no confusion.
\section{Nowcasting and Forecasting Methods}
In this section, we provide the observation model for the stochastic SEIR dynamic system (\ref{stochastic SEIRC1})-(\ref{stochastic SEIRC6}) as well as methods for nowcasting, forecasting and the selection criterion of the population size $N$.
\subsection{Continuous-discrete Dynamic System}
On the basis of the state $\vx(t)$ defined in the model (\ref{stochastic SEIRC1})-(\ref{stochastic SEIRC6}) and the prevalence data that can be obtained, we model the observation equation as follows:
\begin{align}
\label{meas}\mathbf{z}_{k}\triangleq&\left(M_k^c,R_k^c,D_k^c\right)^T= \mH \vx(t_k)+\vv_{k},
\end{align}
where $\vz_k$ is, the measurement of the system state $\vx(t_k)$, which includes the number of the cumulative confirmed infections $M_k^c$, cumulative recovered cases $R^c_k$ and cumulative disease-caused death cases $D^c_k$ at time $t_k$. The statistical error associated with the measurements is modeled as a Gaussian white measurement noise $\vv_{k}$ with a small covariance $\mR_{k}$, and $\mH$ is a linear measurement matrix defined as follows:
\begin{align*}
\mH&=\left(      
\begin{array}{ccccccccc}  
0 & 0 & 1 & 1 & 1 & 0 & 0 & 0 & 0\\
0 & 0 & 0 & 1 & 0 & 0 & 0 & 0 & 0\\
0 & 0 & 0 & 0 & 1 & 0 & 0 & 0 & 0\\
\end{array}
\right) .
\end{align*}
As a consequence, we have the following continuous-discrete hidden Markov dynamic system:
\begin{align}
\label{cd1}\frac{\mathrm{d}\vx(t)}{\mathrm{d}t}&=f(\vx(t),\vw_{\vx}(t)),\\
\label{cd2}\vz_{k}&=\mH\vx(t_k)+\vv_{k},
\end{align}
where $f(\cdot)$ is a simplified expression of the drift function corresponding to the system (\ref{stochastic SEIRC1})-(\ref{stochastic SEIRC6}), and the process noise $\vw_{\vx}(t)=\left(\vw(t)^T, \vw_{\vp}(t)^T\right)^T$ has a covariance $\mQ(t)$ as follows:
\begin{align*}
\mQ(t)=\left(\begin{array}{cc}
\mQ_{\vw}(t)&\\
&\mQ_{\vp}(t)
\end{array}\right),
\end{align*}
where the process noise $\vw_{\vx}(t)$, measurement noise $\vv_k$, and initial state are assumed to be mutually independent.
\subsection{Nowcasting and Forecasting}\label{ss}
For the nonlinear system (\ref{cd1})-(\ref{cd2}), the optimal Bayesian state estimation is usually intractable. There are several strategies to approximate the optimal
estimation, such as the unscented Kalman filter (UKF) \cite{UKF}, \cite{cd-UKF}, and particle filter \cite{particle}. For the system under consideration, we employ the continuous-discrete UKF \cite{UKF}, \cite{cd-UKF} to estimate the random state recursively by the daily data $\vz_k$. Specific steps of the continuous-discrete UKF are shown in Appendix 7.1, where the continuous state model is implemented by the fourth order Runge-Kutta method \cite{RK}. 

\begin{remark}
	Since there are sudden changes in the states and parameters, e.g., the scheme to designate the confirmed infections in China was revised on February 12, 2020 so that the number of cumulative confirmed infections had a drastic increase of 15152 persons \cite{NHC1}. The transfer rate $\beta(t)$ should be modeled differently before and after February 12, 2020. 　Moreover, the classical interacting multiple model (IMM) filter in maneuvering target tracking \cite{mazor1998interacting} can be used to enhance the stability of the algorithm and improve the performance  of the nowcasting and forecasting methods.
\end{remark}

Furthermore, the prediction of the state and covariance are derived by the model prediction step of UKF through the model (\ref{stochastic SEIRC1})-(\ref{stochastic SEIRC6}) and smooth parameters. Let $\vm_{i+j|i}=\left(I_{i+j|i}, R_{i+j|i}, D_{i+j|i}\right)^T$ denote the vector consisting of the prediction of the confirmed infections, recovered cases and disease-caused death cases from time $t_i$ to $t_{i+j}$, where $M^c_{i+j|i}=I_{i+j|i}+R_{i+j|i}+D_{i+j|i}$ represents the prediction of the cumulative confirmed infections from time $t_i$ to $t_{i+j}$, and $\sigma_{i+j|i}$ is the standard deviation of the error of $M^c_{i+j|i}$. Then, the approximate 95\% confidence intervals (CIs) are calculated as follows:
\begin{align}
\left(M^c_{i+j|i}-2\sigma_{i+j|i},\ M^c_{i+j|i}+2\sigma_{i+j|i}\right).
\end{align}

Based on the proposed stochastic SEIR model under public health interventions (\ref{stochastic SEIRC1})-(\ref{stochastic SEIRC6}), the nowcasting and forecasting methods described above, the local and global spread trends of the epidemic are analyzed under public health interventions and different control intensities, respectively.
\subsection{Selection Criterion of the Population Size $N$}
When the epidemic spreads freely in a city, the urban population can be chosen as the population size $N$. However, under public health interventions, the value of $N$ should be selected carefully, and we use $N$ to measure the intensity of interventions. Note that the $E, I, R, D$ are integer valued in practical applications, based on the recursive method, Propositions \ref{pro0}-\ref{pro} (see Appendix 7.3-7.4) show that $N$ cannot be estimated since the state estimates are the same for different $N$ when they are large. Thus, it is unnecessary to consider a very large $N$, especially under strict control measures or early on in an outbreak. On the other hand, it is not difficult to show that $N$ can be estimated when it is not too large since $\frac{E\left(E+I+R+D\right)}{N}$ is not infinitesimal (see the analysis in Appendix 7.2). Moreover, we provide a selection criterion for the value of $N$ by solving the following optimization problem, i.e., we select the population size by minimizing the average relative prediction error of the cumulative confirmed infections in the next $J$ time steps:
\begin{align}\label{criterion2}
\begin{split}
\mathop{\min}_{N}\ &\frac{1}{J}\sum_{j=1}^{J}  e_j(N)
\\
\text{s.t.}\ &N\leq N_u,
\end{split}
\end{align}
where $J$ is a constant, e.g., 7 days. $N$ is the variable over which optimization is to be carried out, $N_u$ may be half of the regional population or other values based on prior data, which depends on the intensity of the control measures etc., and the $e_j(N)$ is calculated as follows:
\begin{align}
e_j(N)=\frac{1}{k-j}\sum_{i=1}^{k-j}\left|\frac{M^c_{i+j|i}(N)-M^c_{i+j}}{M^c_{i+j}}\right|,
\end{align}
where $M^c_{i+j}$ and $M^c_{i+j|i}(N)$ are the observations of the cumulative confirmed infections and the prediction of the cumulative confirmed infections based on the population size $N$ through time $t_i$ to $t_{i+j}$ for $j=1, \ldots, J$, respectively.
\section{Data Analysis of COVID-19 Epidemic}
Based on the proposed stochastic SEIR model under public health interventions (\ref{stochastic SEIRC1})-(\ref{stochastic SEIRC6}) and the nowcasting and forecasting methods, the local and global spread trends of COVID-19 epidemic are analyzed under different control intensities, respectively. All the data are collected from daily reports of National Health Commission of the People’s Republic of China (NHC) and Tencent News real-time epidemic tracking \cite{NHC}, \cite{Tencent}, including the number of cumulative confirmed infections, recovered cases and disease-caused death cases. 

The initial prior ranges of the states and parameters for the stochastic SEIR model (\ref{stochastic SEIRC1})-(\ref{stochastic SEIRC6}) are as follows: 
\begin{itemize}
	\item $E, I, R, D$: the initial values of $I, R, D$ can be set as the number of reported confirmed cases, recovery cases and disease-caused death cases, respectively. The initial number of latent infections $E$ can be set as $c\cdot (I+R+D)$, where $5\leq c\leq20$.
	\item $\alpha$: the contact rate. $0.1\leq\alpha\leq0.3$.
	\item $\beta$:  the transfer rate from latent to confirmed. 7 days $\leq\frac{1}{\beta}\leq$ 14 days.
	\item $\gamma^1, \gamma^2$: the recovery rate and disease-caused death rate. $0.05\leq\gamma^1\leq0.1$, $0.001\leq\gamma^2\leq0.01$.
\end{itemize}

\subsection{The Selection of the Population Size $N$ and Estimation of the Basic Reproductive Number}

Fig. \ref{fig_hubei_trend2} shows the prediction of the confirmed infections $I(t)$ based on the data through January 20 to February 24 of Hubei Province, China under different intensities of  public health interventions, i.e, maintaining the control intensity and the population size $N=1\times10^6$, relaxing the control intensity and the population sizes $N=2\times10^6$, and $N=3\times10^6$, respectivley. The results indicates that relaxing the Hubei quarantine would lead to a second epidemic peak. It also indicate that the selection of the population size for pathogen transmission is a key factor for the forecasting of the epidemic and $N$ is identifiable for this case. 

Fig. \ref{fig_beijing_trend1} shows the prediction for Beijing City. It indicates that when $N$ is greater than $10^5$, the predicted values are the same, i.e, $N$ is unidentifiable. This is consistent with Propositions \ref{pro0}-\ref{pro}. In fact, early on in the epidemic, Beijing took very strict control measures so that the basic reproductive number is very small, the small number of infections cannot spread to a larger population.
\begin{figure}[H]
	\centering
    \vspace{0.35cm} 
	\subfigure[Hubei province, China]{
		\includegraphics[width=0.45\linewidth]{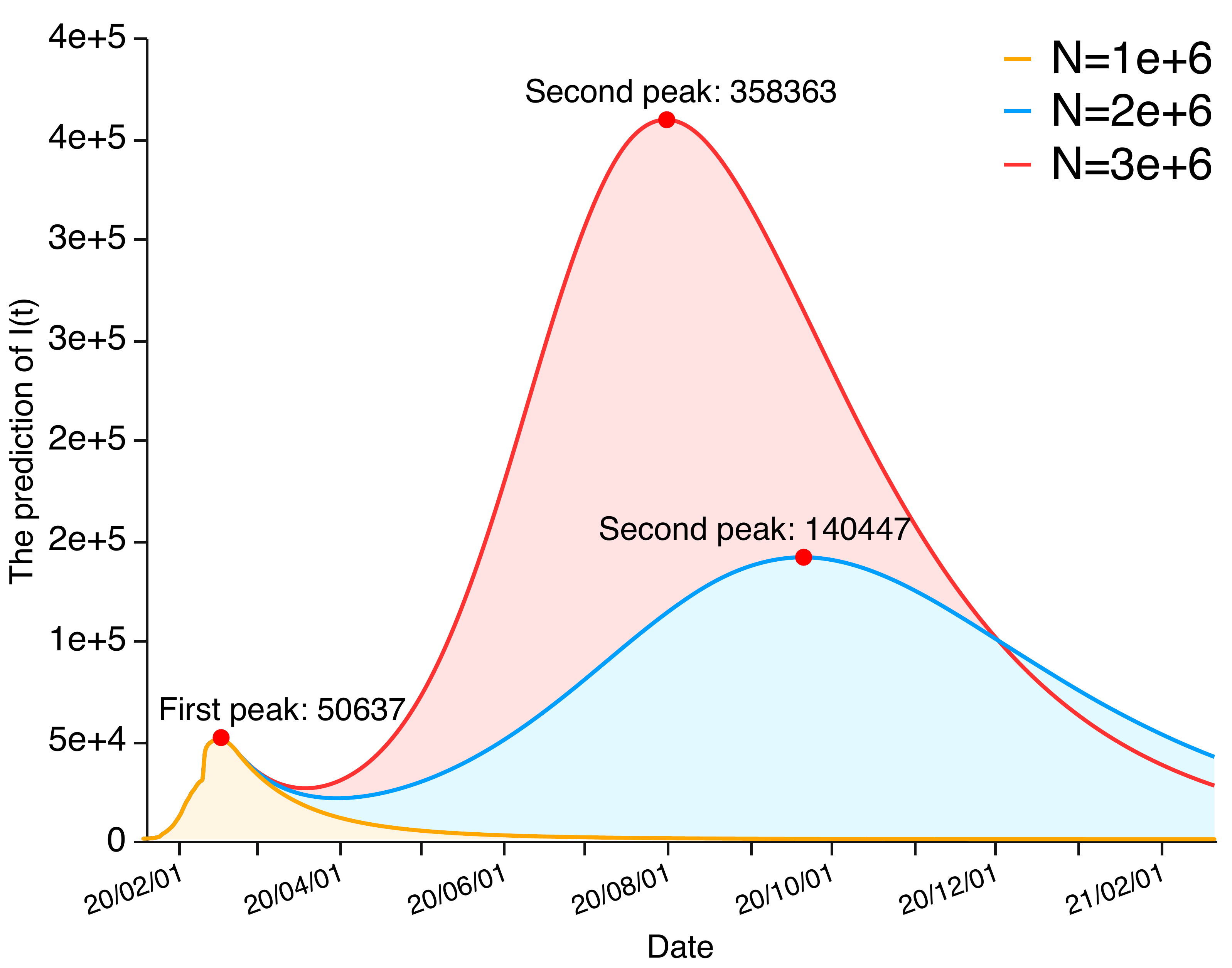}
		\label{fig_hubei_trend2}
	}
	\subfigure[Beijing City, China]{
		\includegraphics[width=0.45\linewidth]{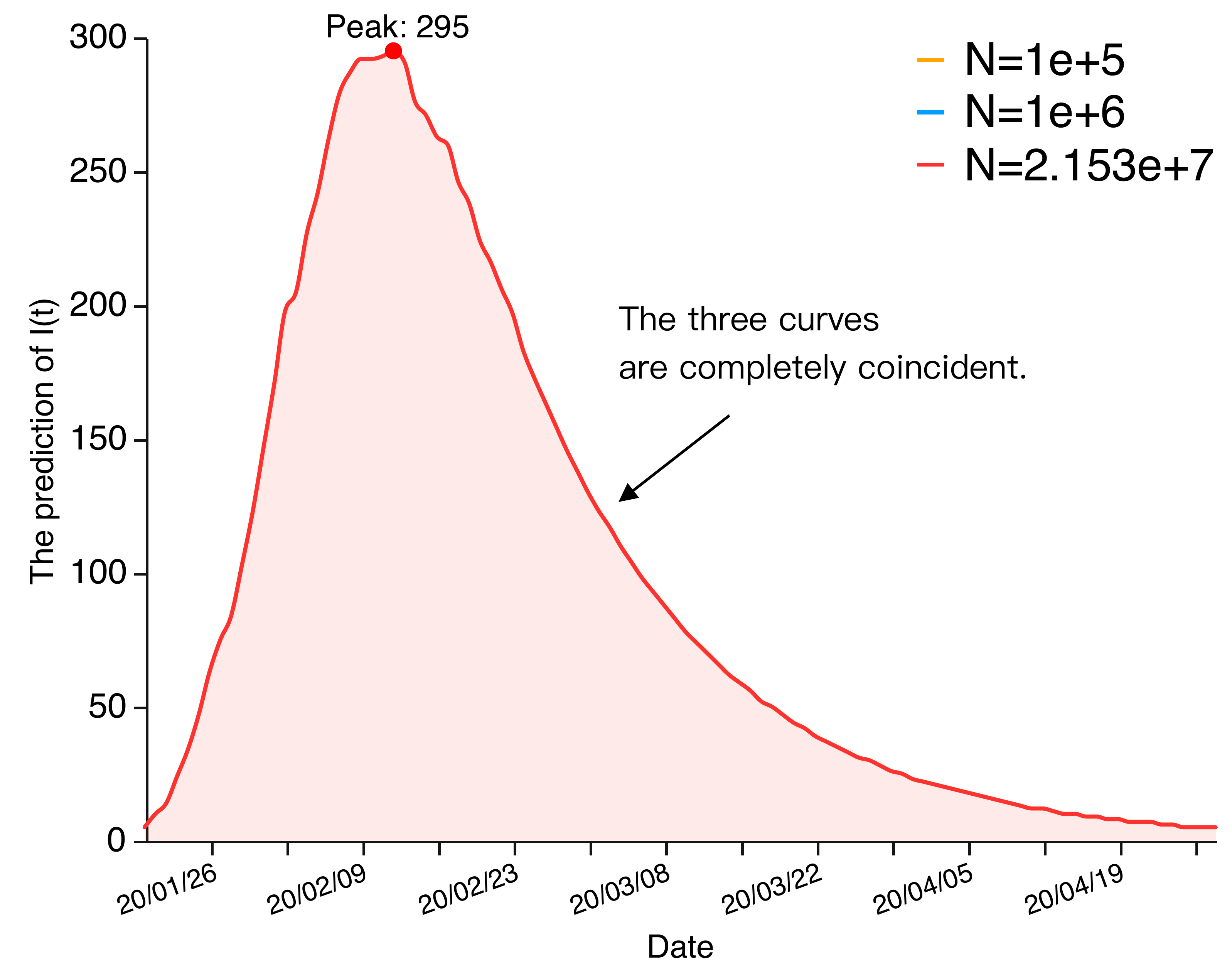}
		\label{fig_beijing_trend1}
	}
	\centering
	\caption{The forecasting of the epidemic trend of $I(t)$ with different population sizes.}
	\label{identifyN}
\end{figure}

Therefore, for global prediction, it is required to select the population size carefully. Through the selection criterion (\ref{criterion2}) and setting $J=7$, we select the population size based on the prevalence data of COVID-19 epidemic \cite{NHC}, \cite{Tencent} from the date of the first confirmed infection to April 2, 2020. 

Table \ref{tab1} gives the selected values of $N$ and national populations \cite{Population} of several major epidemic affected countries and GEC. Here, we restate that $N$ is the population size for pathogen transmission. The total national population is divided into the population for pathogen transmission and the isolated uninfected population. It shows that the selected $N$ under public health interventions are far less than the national population. Comparing Iran with Italy, they have a similar total population, however, Iran has a smaller size $N$ for pathogen transmission, i.e, stronger control measures \cite{iran}.

\begin{table}[H] 
	\centering
	\caption{\label{tab1}{The selected $N$ and national population}}
	\begin{tabular}{ccccccccccccc} 
		\toprule 
		Region  & USA & Austria &China  &England  & France & Germany
		\\ 
		\midrule 
		Population  & 3.2717e+8 &8.8223e+6 &13.9273e+8 &6.6274e+7 & 6.4769e+7&  8.2792e+7
		\\	
		Selected $N$ & 2.36e+6 &4.8e+5& 1.22e+6 &1.25e+6	&1.68e+6&1.86e+6
		\\
		\midrule 
		Region  & Iran & Italy & Japan& Netherlands& Spain & GEC
		\\
		\midrule
		Population  &8.2084e7  &6.0484e+7 &1.2653e+8&1.7181e+7 & 4.6733e+7 &6.2013e+09
		\\
		Selected $N$ &9.5e+5& 2.08e+6&5.6e+5&4.2e+5 &1.82e+6&1.6e+7
		\\
		\bottomrule 
	\end{tabular} 
\end{table}
Moreover, to validate the selected population sizes $N$, Fig. \ref{beijing_modelfit} shows the reported confirmed infections from the date of the first confirmed case to April 2, 2020, daily prediction of one time step and the corresponding 95\% CIs for several major epidemic affected countries and GEC, respectively. It indicates that the 95\% CIs can cover the true confirmed infections, no matter what in the outbreak phase (see Fig. \ref{Austria1}, \ref{England1}-\ref{globalcurve}), or end of the outbreak (see Fig. \ref{Chinacurve}).

Specifically, it's worth noting that in Fig. \ref{Chinacurve} the number of cumulative confirmed infections has a sudden increase owing to the revision of the diagnosis scheme in China on February 12, 2020 \cite{NHC1}, the stochastic SEIR  model under public health interventions (\ref{stochastic SEIRC1})-(\ref{stochastic SEIRC6}) and the UKF method can quickly track the state in several time steps, and then predict the state with a small CI. Besides, Fig. \ref{globalcurve} indicates the epidemic trend of COVID-19 globally is still in the outbreak with exponential growth. Furthermore, more prediction results of 1-7 days are updated and available on https://github.com/SCU-Bigdata/2019-nCoV-Forecast in time since February 7, 2020.
\begin{figure}[H]
	\centering  
	\vspace{-0.35cm} 
	\subfigtopskip=2pt 
	\subfigbottomskip=2pt 
	\subfigcapskip=-5pt 
	
	\subfigure[Austria]{
		\includegraphics[width=0.3\linewidth]{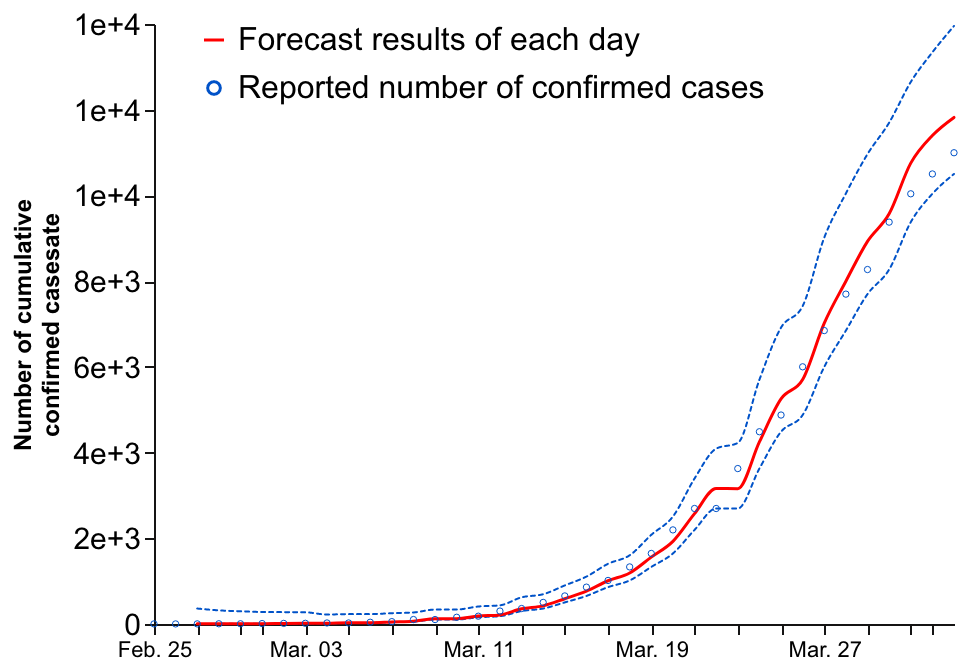}
		\label{Austria1} }%
	\subfigure[China]{
		\includegraphics[width=0.3\linewidth]{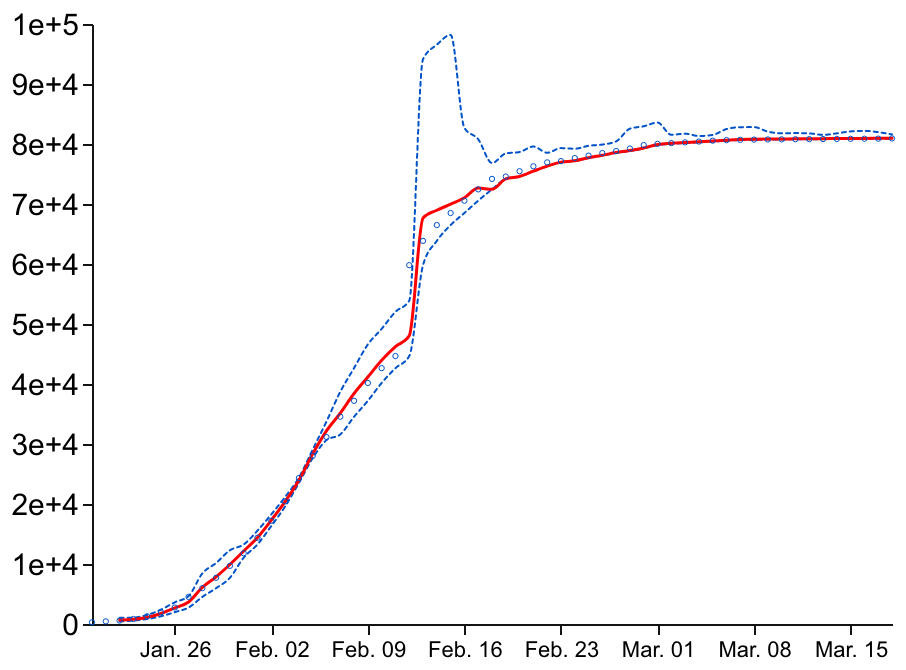}
		\label{Chinacurve}
	}
	\subfigure[England]{
		\includegraphics[width=0.3\linewidth]{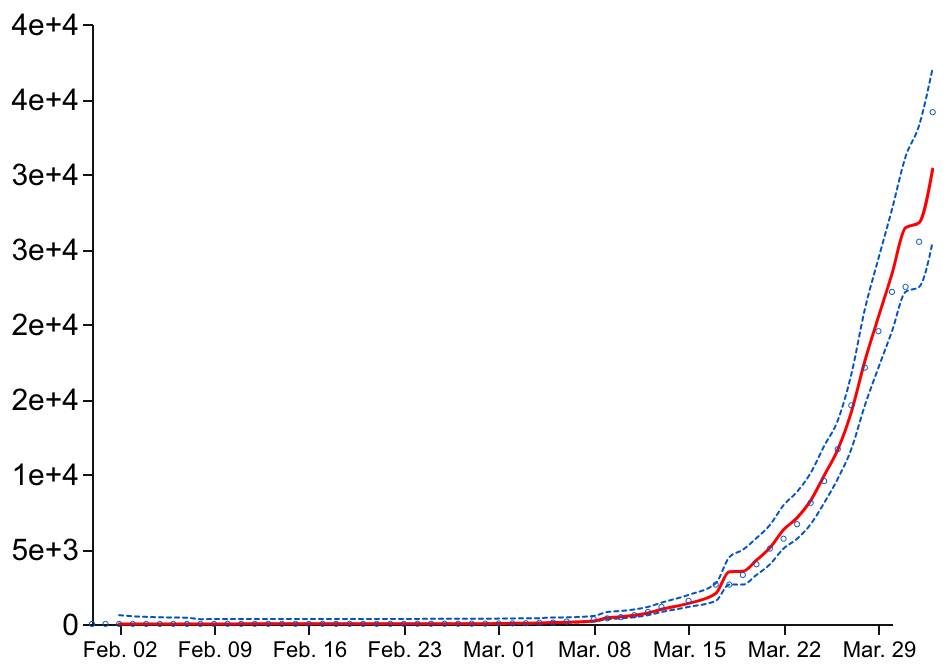}
		\label{England1} }
	
	\subfigure[France]{
		\includegraphics[width=0.3\linewidth]{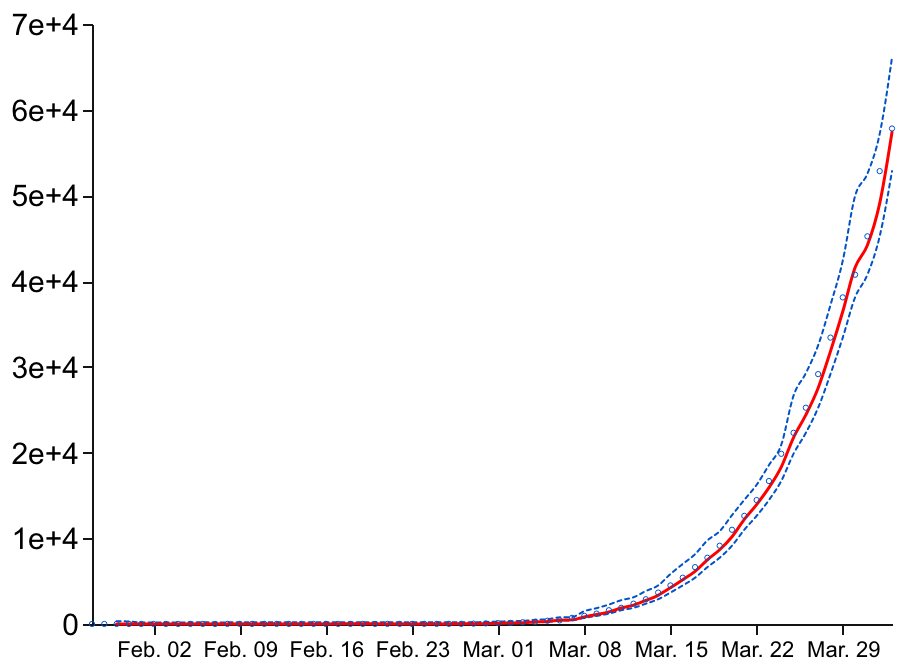}
	}
	\subfigure[Germany]{
		\includegraphics[width=0.3\linewidth]{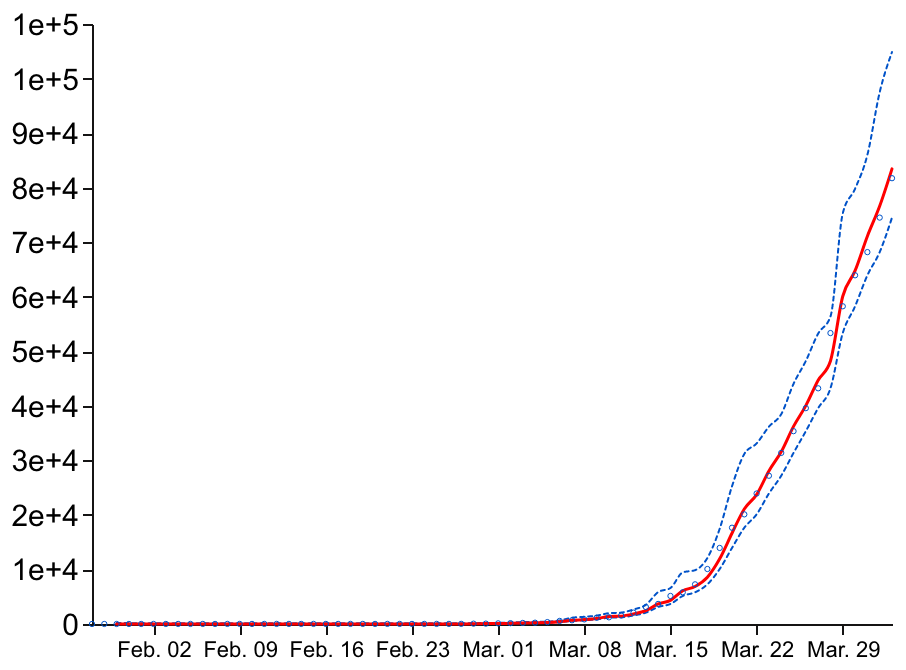}
	}%
	\subfigure[Italy]{
		\includegraphics[width=0.3\linewidth]{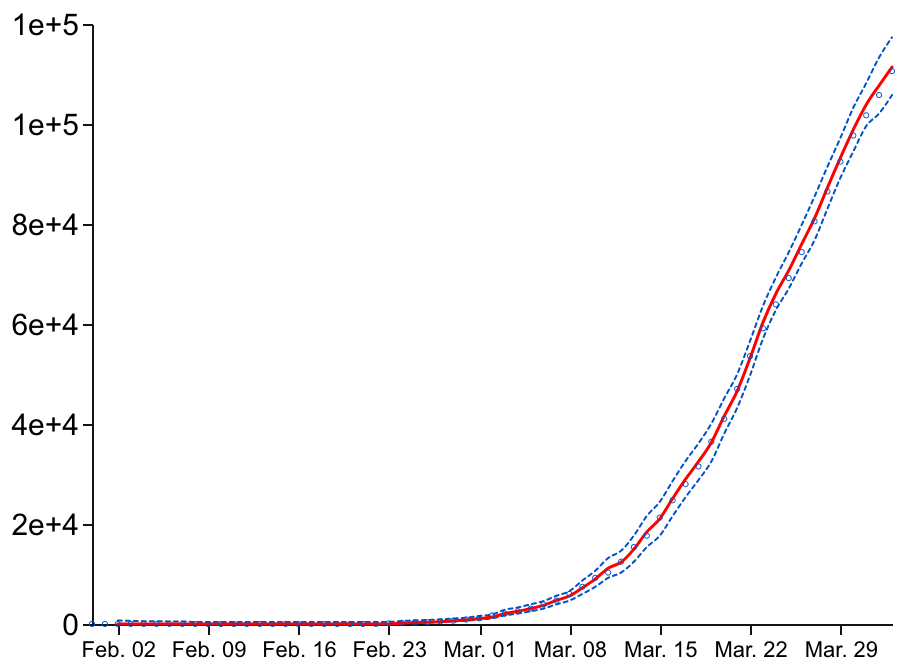}
		\label{Italy1}}%
	
	\subfigure[Iran]{
		\includegraphics[width=0.3\linewidth]{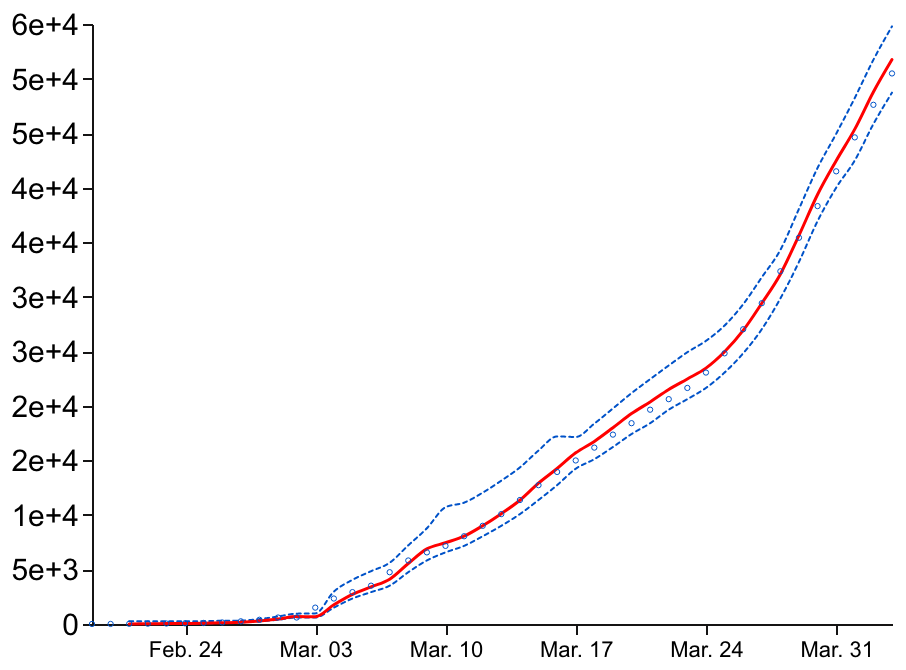}
		\label{B-iran1}}%
	\subfigure[Japan]{
		\includegraphics[width=0.3\linewidth]{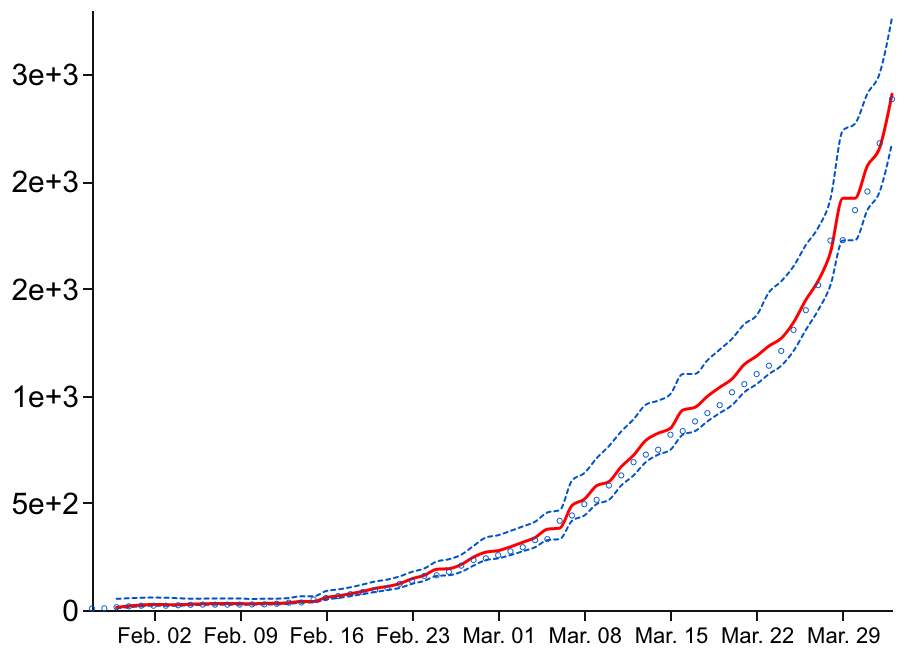}
	}%
	\subfigure[Netherlands]{
		\includegraphics[width=0.3\linewidth]{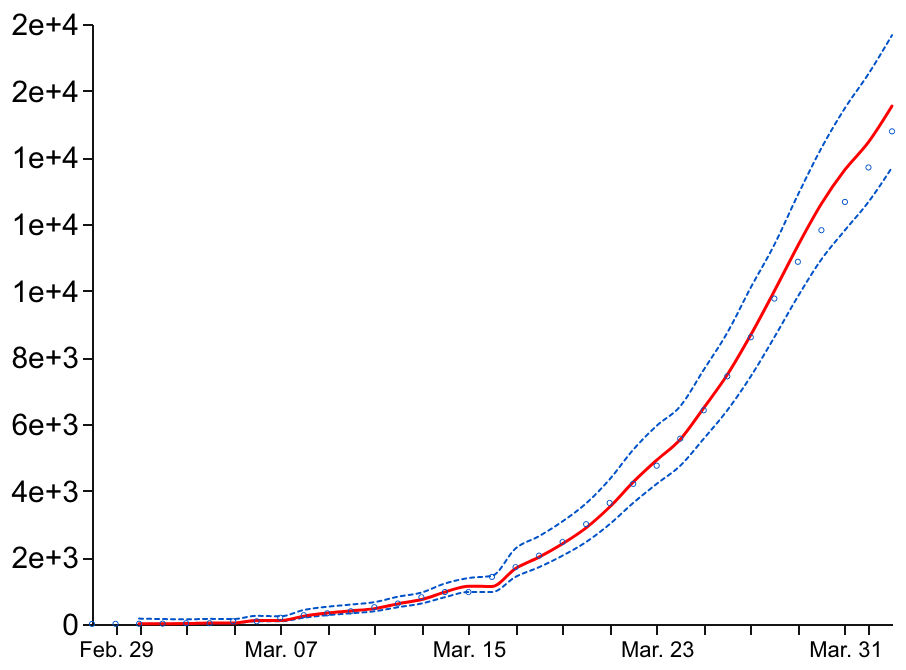}
		\label{Koreacurve}
	}%

	\subfigure[Spain]{
		\includegraphics[width=0.3\linewidth]{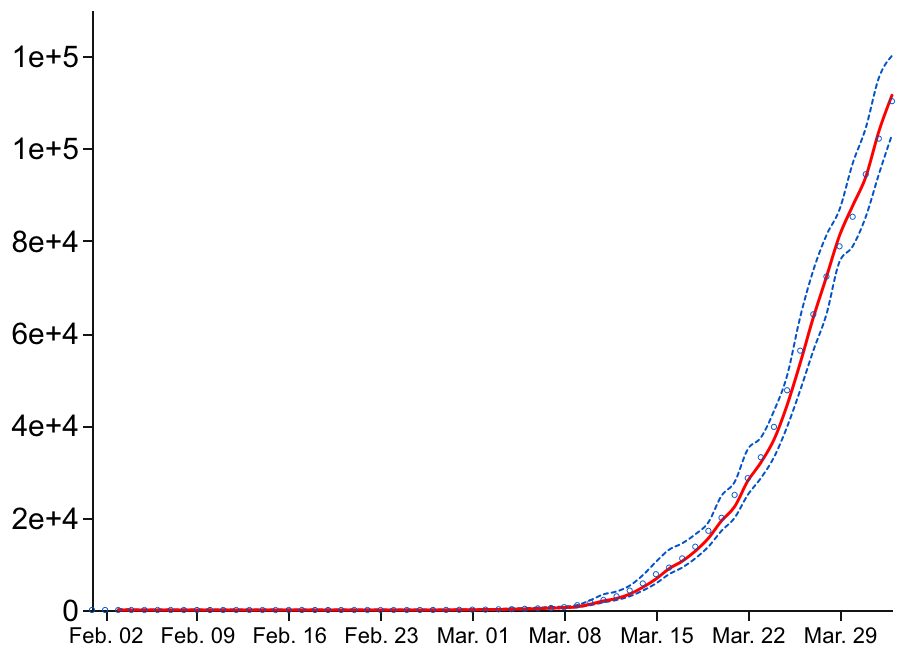}
		\label{Spain1}}%
		\subfigure[USA]{
		\includegraphics[width=0.3\linewidth]{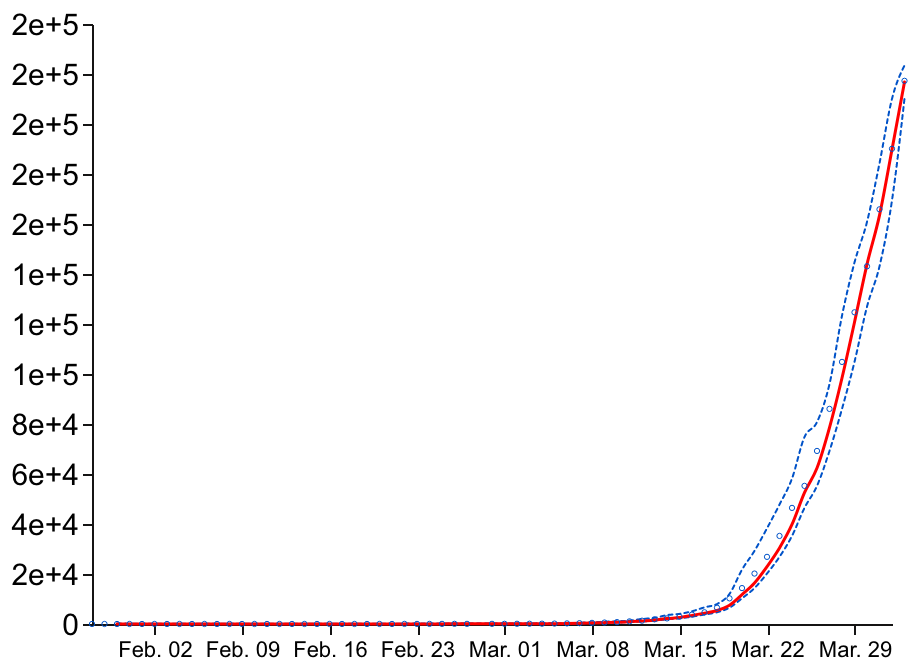}\label{America1}}
	\subfigure[GEC]{
		\includegraphics[width=0.3\linewidth]{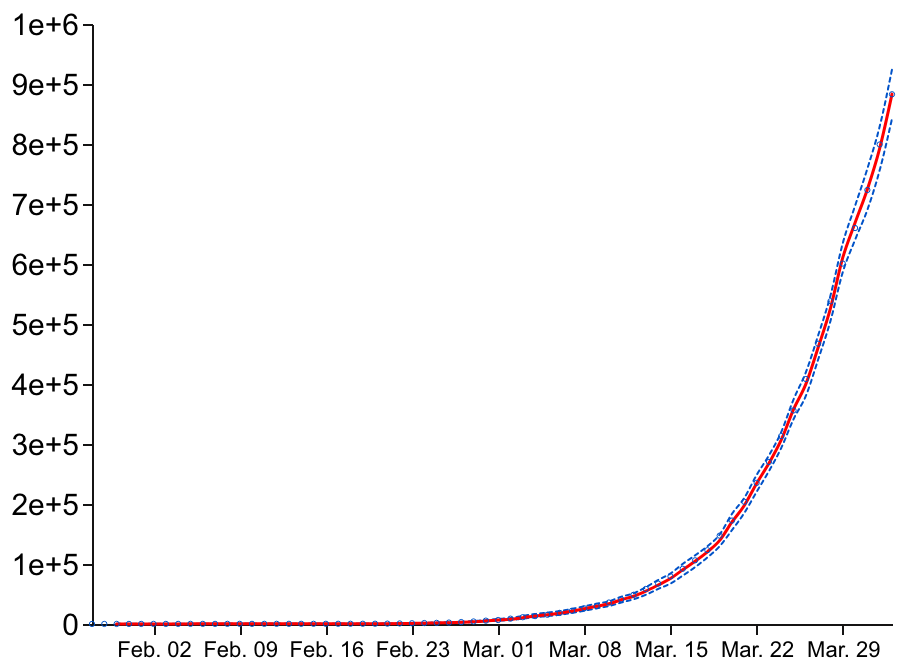}\label{globalcurve}}	
	\centering
	\caption{Daily prediction of one time step and 95\% CI derived by the model (\ref{stochastic SEIRC1})-(\ref{stochastic SEIRC6}) though continuous-discrete UKF. }
	\centering
	\label{beijing_modelfit}
\end{figure}

Fig. \ref{R01} shows the estimation of the basic reproductive number $R_0$ of China from January 21, 2020 to March 18, 2020 when there were no local newly confirmed cases \cite{NHC}. It shows that the basic reproductive number increases monotonically early on in the outbreak, then monotonically decreases with a sharp infection point. This is consistent with the policy decision in which Chinese government implemented unprecedented public health interventions to prevent the spread of the epidemic and  the incubation period is about 7-14 days \cite{zhong}. The reason for monotonic increase in the early stages of the epidemic may be that, in a limited local city, the density of infections increases as a result of the infections increasing exponentially. It results in a quick growth of the contact rate $\alpha(t)$. The recovery rate is very low due to the recovery cycle. The reason for monotonic decrease is that the public health interventions (e.g. scales of quarantine, strict controls on travel, hand hygiene, and use of face masks etc. \cite{wuhancontrol}) significantly decrease  the contact rate $\alpha(t)$ and the recovery rate gradually increases, which results in a monotonic decrease of $R_0$. In Fig. \ref{R02}, although the basic reproductive number of GEC has passed the infection point, which is consistent with the policy decision in which European countries and USA have implemented strict public health interventions, the current $R_0$ is still high and fluctuates around 3.64. 
\begin{figure}[H]
	\centering
		\vspace{0.35cm} 
	\subfigure[China]{
		\includegraphics[width=0.45\linewidth]{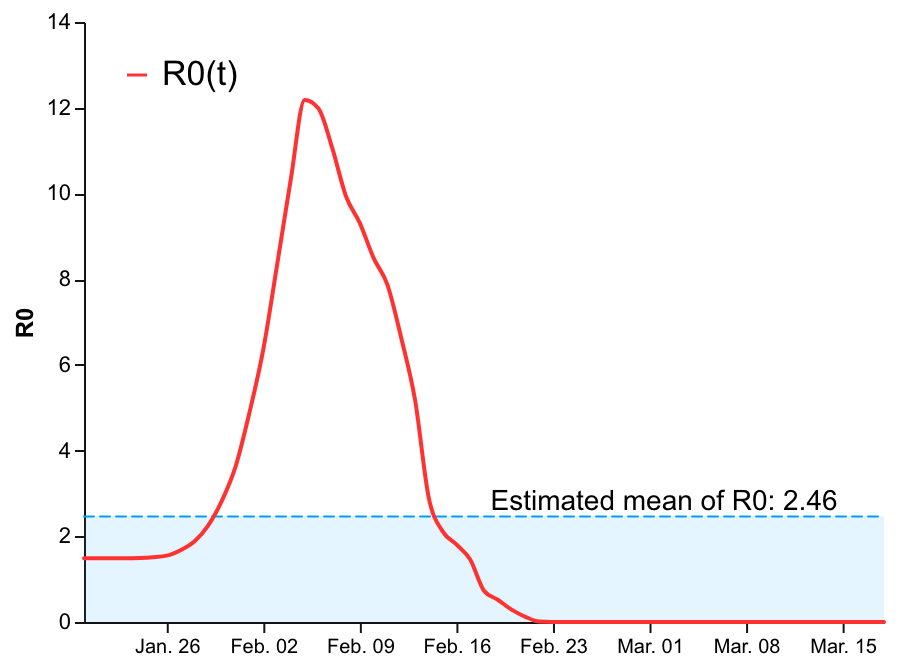}
		\label{R01}
	}
	\subfigure[GEC]{
		\includegraphics[width=0.45\linewidth]{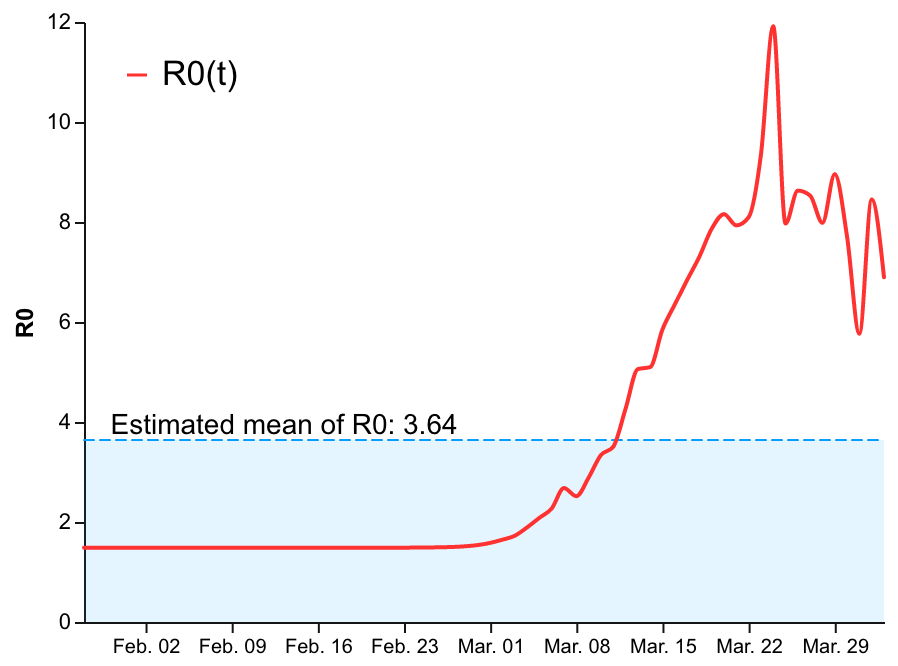}
		\label{R02}
	}
	\centering
	\caption{Estimation results for the basic reproductive number.}
	\label{R00}
\end{figure}

We also present the estimated mean values and 95\% CI of the key epidemiological parameters of GEC and China in Table \ref{tab2}. The mean of the basic reproductive number of COVID-19 in China is 2.46 (95\% CI: 2.41-2.51). This finding aligns with other recent estimates of the basic reproductive number for this time period \cite{SEIR_cite3, du2020risk, R01, R02, riou2020pattern, imai2020report}. 
\begin{table}[H]
	\centering
	\caption{\label{tab2}{The estimated mean and 95\% CI of parameters}}
	\begin{tabular}{lcc} 
		\toprule 
		Parameters  & GEC & China
		\\ 
		\midrule 
		Contact rate	$\alpha$ &0.17 $\left(0.12, 0.22\right)$ &0.07 $\left(0, 0.23\right)$
		\\
		Transfer rate $\beta$ &0.12  $\left(0.07, 0.17\right)$ &0.06  $\left(0, 0.18\right)$
		\\
		Remove rate	$r^1+r^2$ &0.07  $\left(0, 0.14\right)$&0.06  $\left(0, 0.12\right)$
		\\
		Basic reproductive number	$R_0$ &3.64 $\left(3.55, 3.72\right)$ &2.46  $\left(2.41, 2.51\right)$
		\\
		\bottomrule
	\end{tabular} 
\end{table}

\subsection{The Nowcasting of Local and Global Infections}
The number of latent infections $E(t)$ can be estimated in time again by UKF nonlinear filtering. Fig. \ref{hidenE} shows the reported cumulative confirmed infections, and the estimate of the actual infections, where the difference between the two lines represents $E(t)$. It indicates that the situation in USA is the most serious one and it has $1.90\times10^5$ (95\% CI: $1.87\times10^5$-$1.92\times10^5$) latent infections on April 2, 2020. In England, the latent infections account for a very large fraction 58\% of the total number of infections. We infer that the outbreak will commence in the next two weeks with a lag time about 7-14 days \cite{zhong}. Similar results are anticipated for Austria, Japan and Netherlands with ratios between 51\%-54\%. Basically, Italy is near the point of inflection. Moreover, the estimated number of latent infections of GEC is about $7.47\times10^5$ (95\% CI: $7.32\times10^5$-$7.62\times10^5$), which is 32\% of the total infections of GEC and about 0.47 times of the confirmed infections.
\begin{figure}[htbp]
	\centering  
	\vspace{-0.35cm} 
	\subfigtopskip=2pt 
	\subfigbottomskip=2pt 
	\subfigcapskip=-5pt 

	\subfigure[Austria]{
		\includegraphics[width=0.3\linewidth]{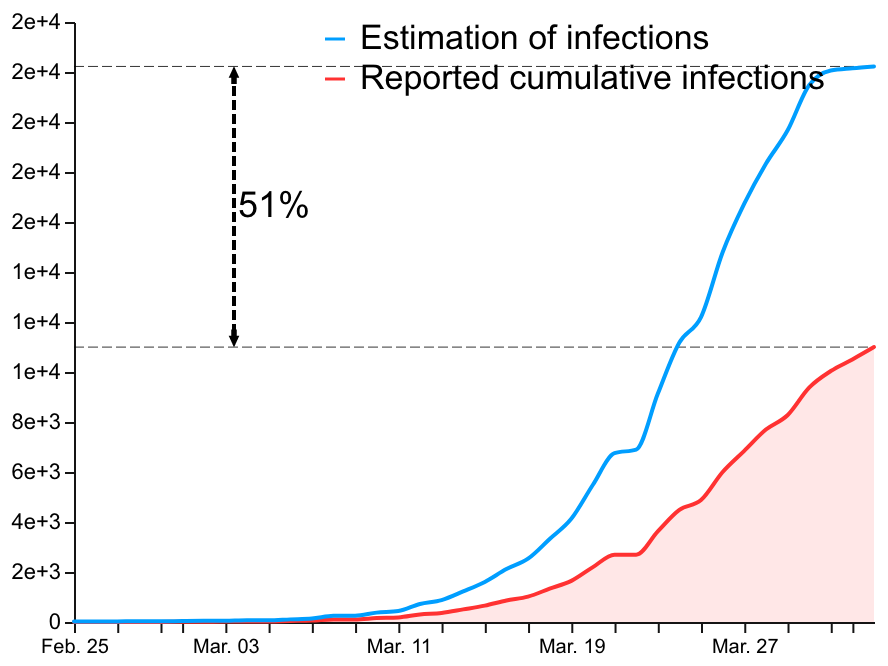}
	}
	\subfigure[England]{
		\includegraphics[width=0.3\linewidth]{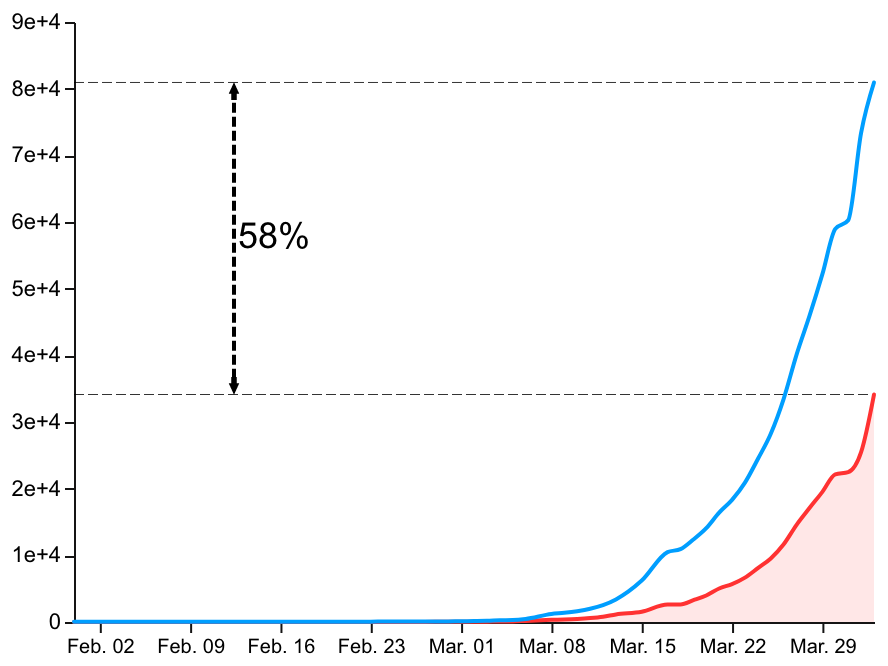}
	}
	\subfigure[France]{
		\includegraphics[width=0.3\linewidth]{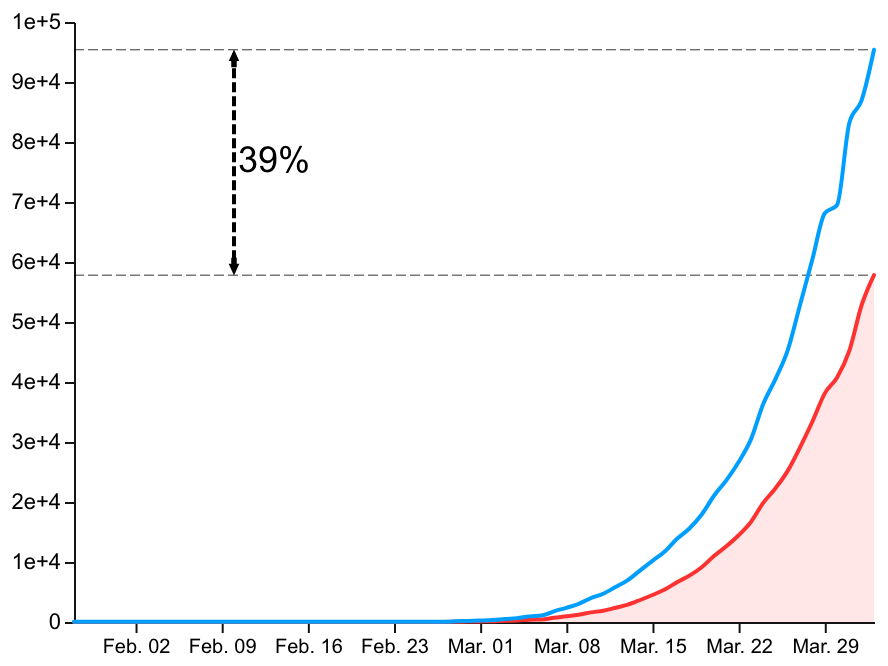}
	}

	\subfigure[Germany]{
		\includegraphics[width=0.3\linewidth]{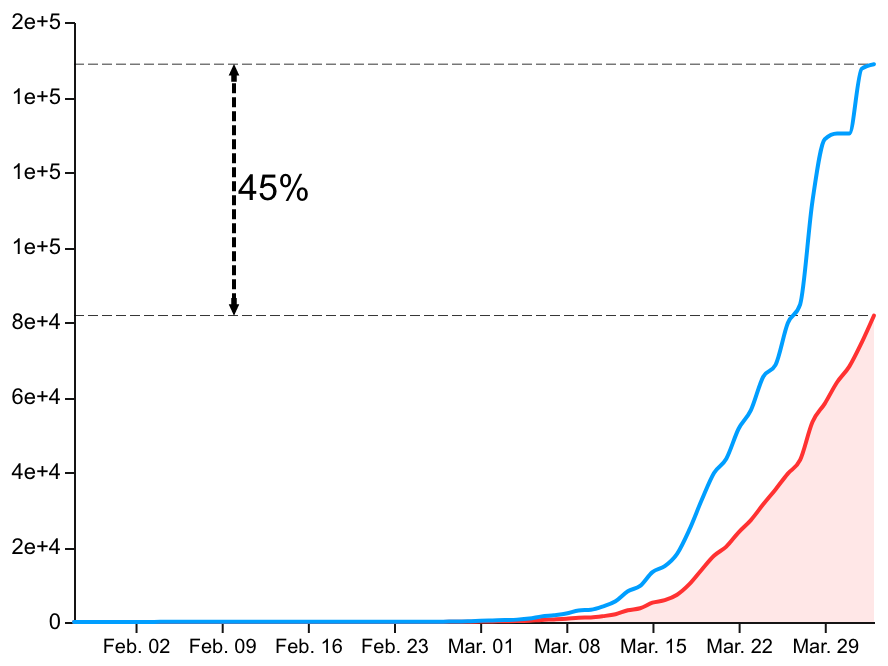}
	}%
	\subfigure[Italy]{
		\includegraphics[width=0.3\linewidth]{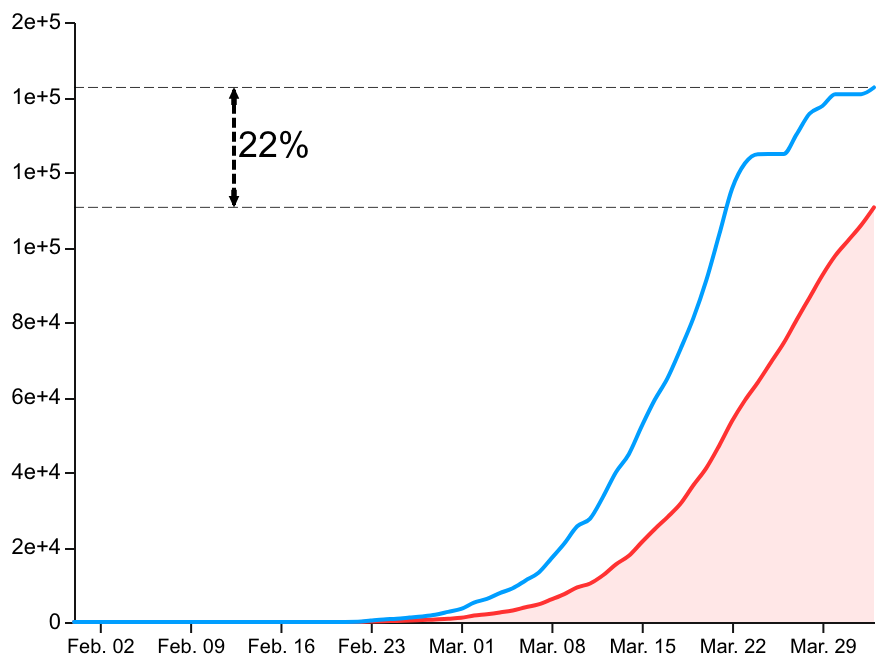}
	}
	\subfigure[Iran]{
		\includegraphics[width=0.3\linewidth]{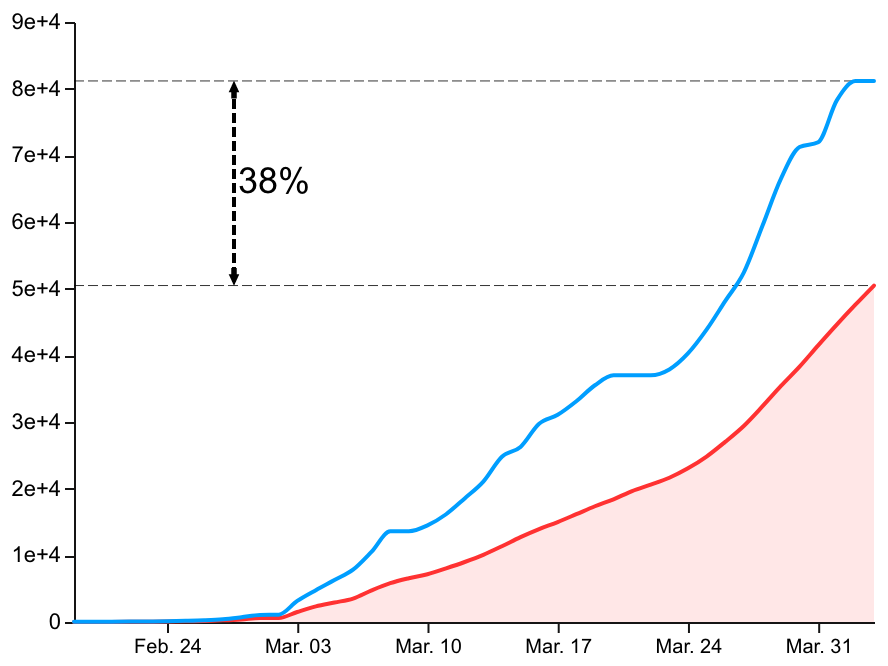}
	}

	\subfigure[Japan]{
		\includegraphics[width=0.3\linewidth]{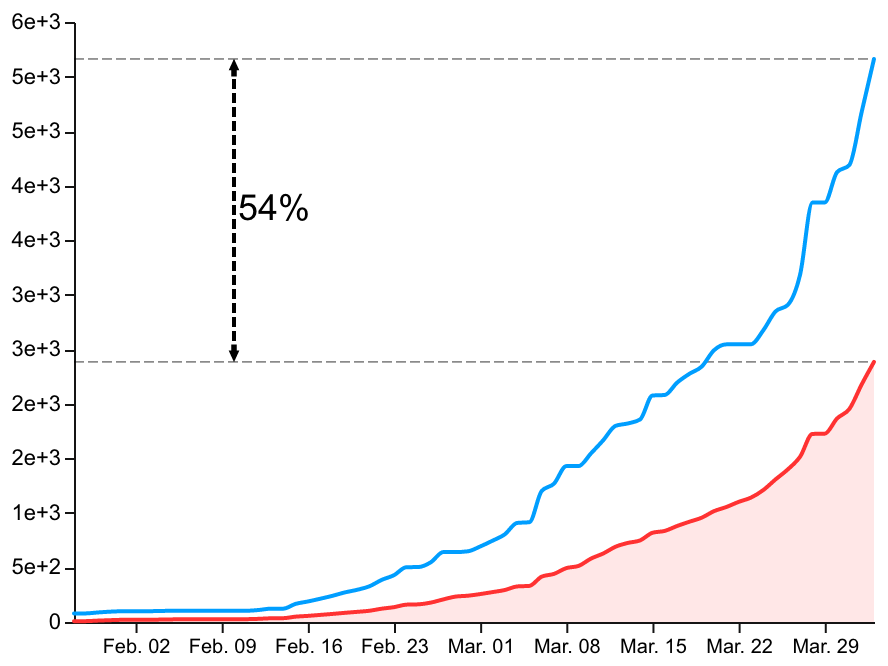}
	}
	\subfigure[Netherlands]{
		\includegraphics[width=0.3\linewidth]{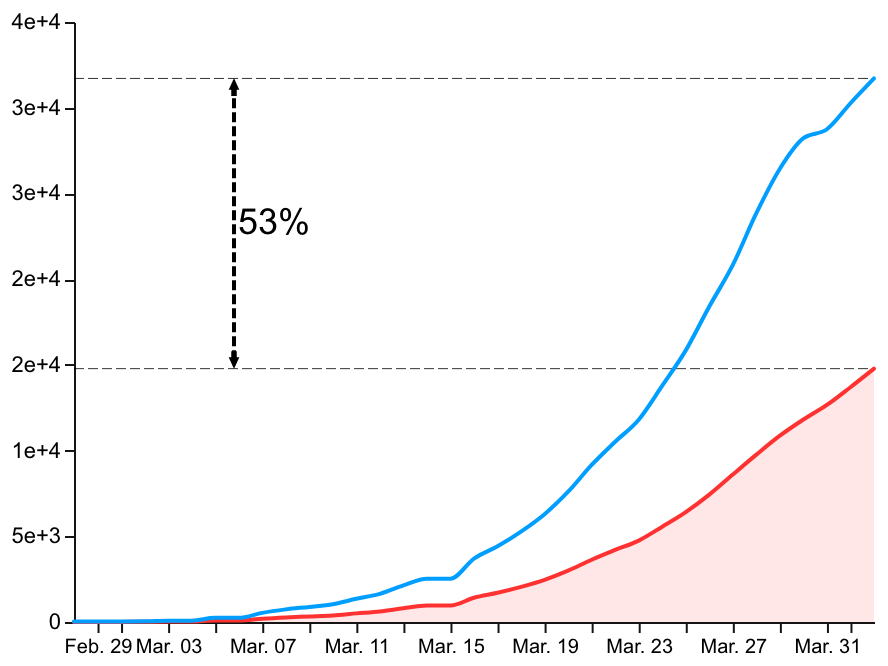}
	}
	\subfigure[Spain]{
		\includegraphics[width=0.3\linewidth]{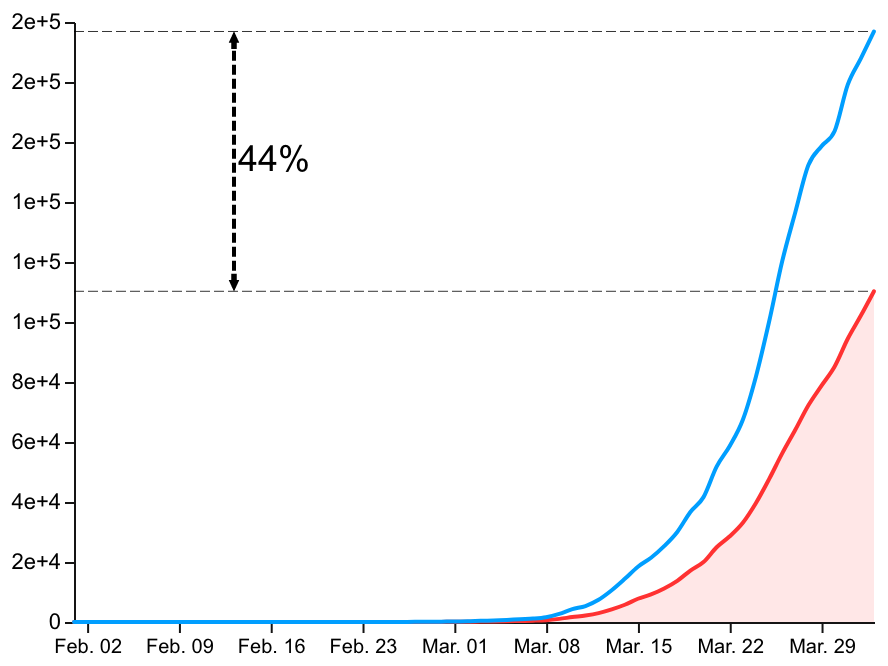}
	}

	\subfigure[Switzerland]{
		\includegraphics[width=0.3\linewidth]{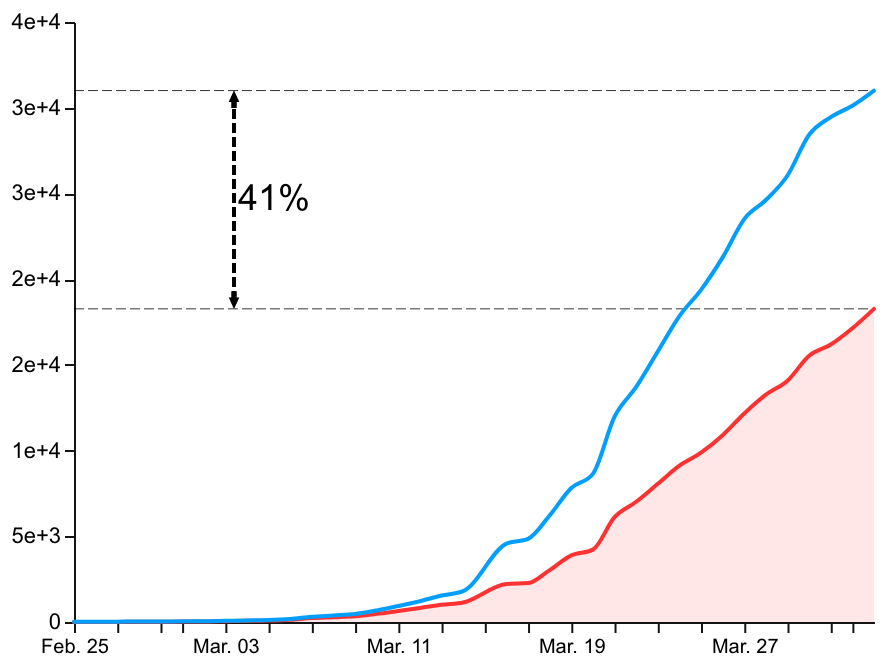}
		\label{KoreaE}
	}
	\subfigure[USA]{
	\includegraphics[width=0.3\linewidth]{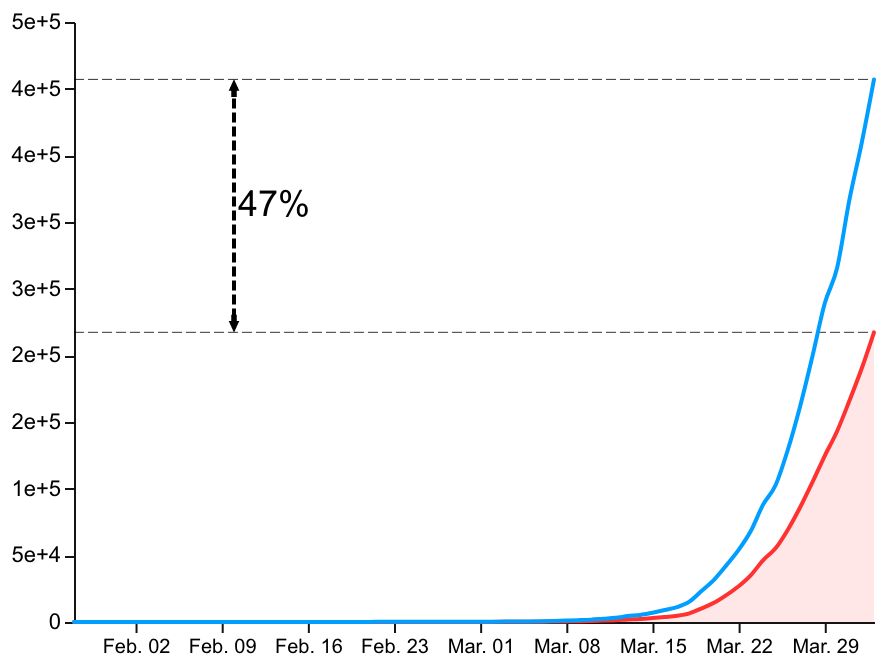}}	
	\subfigure[GEC]{
		\includegraphics[width=0.3\linewidth]{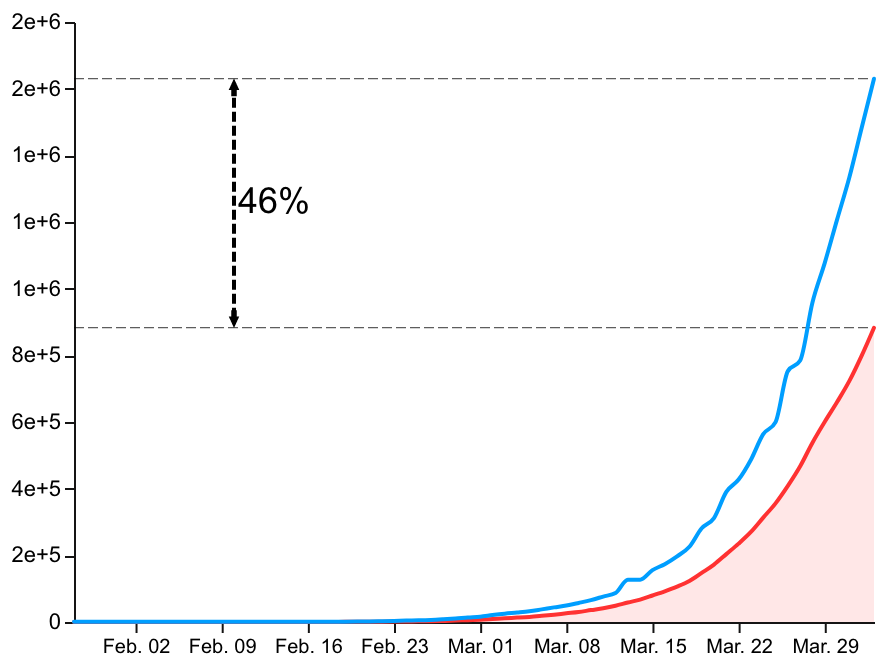}
	}
	\centering
	\caption{The reported number of cumulative confirmed infections and the estimate of the total number of infections.}
	\label{hidenE}
\end{figure}

Herein, we estimated the latent infections in 109 major epidemic affected countries worldwide except China where the reported infections are over 100 cases. Fig. \ref{first_confirm111} shows the estimate of the latent infections $E(t)$ on April 2, 2020. It indicates that there are 61 countries, 14 countries and one country with more than 1,000, 10,000, 100,000 latent infections, respectively. USA is the most serious one as it has $1.90\times10^5$ (95\% CI: $1.87\times10^5$-$1.92\times10^5$) latent infections. Germany and Spain have $6.70\times10^4$ (95\% CI: $6.55\times10^4$-$6.86\times10^4$), $8.66\times10^4$ (95\% CI: $8.59\times10^4$-$8.73\times10^4$) latent infections, respectively. The number of latent infections of Austria, Belgium, Canada, England, France, Italy, Iran, Netherlands, Portugal and Switzerland and Turkey are between 10,000 and 50,000.

\begin{figure}[htbp]
	\centering  
	\vspace{-0.35cm} 
	\subfigtopskip=2pt 
	\includegraphics[scale=0.35]{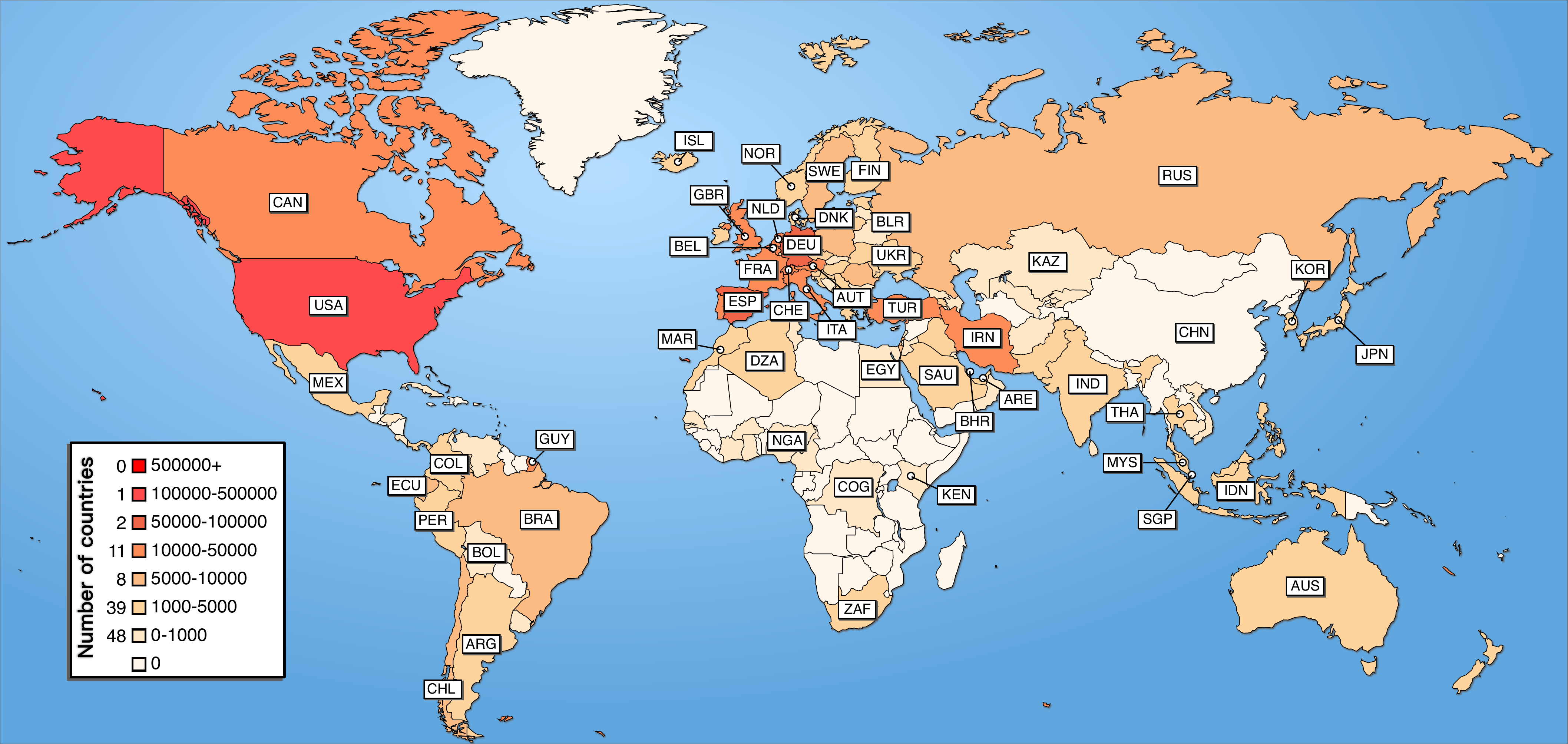}
	\caption{The estimates of the latent infections $E(t)$ in several major epidemic affected countries as of April 2, 2020.}
	\label{first_confirm111}
\end{figure}

\subsection{The Forecasting of the Local and Global Infections}
To compare the global trend for two weeks later, Fig. \ref{current} shows the confirmed infections $I(t)$ around the globe. There are 46 countries, 11 countries and one country where the number of confirmed infections are greater than 1,000, 10,000 and 100,000, respectively. The number of confirmed infections $I(t)$ of USA, Germany, Italy and Spain are about 203,402, 61,556, 80,572 and 73,492, respectively. Some other countries such as Belgium, England, France, Iran, Netherlands, Switzerland and Turkey are between 10,000 and 50,000.
\begin{figure}[htbp]
	\centering  
	\vspace{-0.35cm} 
	\subfigtopskip=2pt 
\includegraphics[scale=0.35]{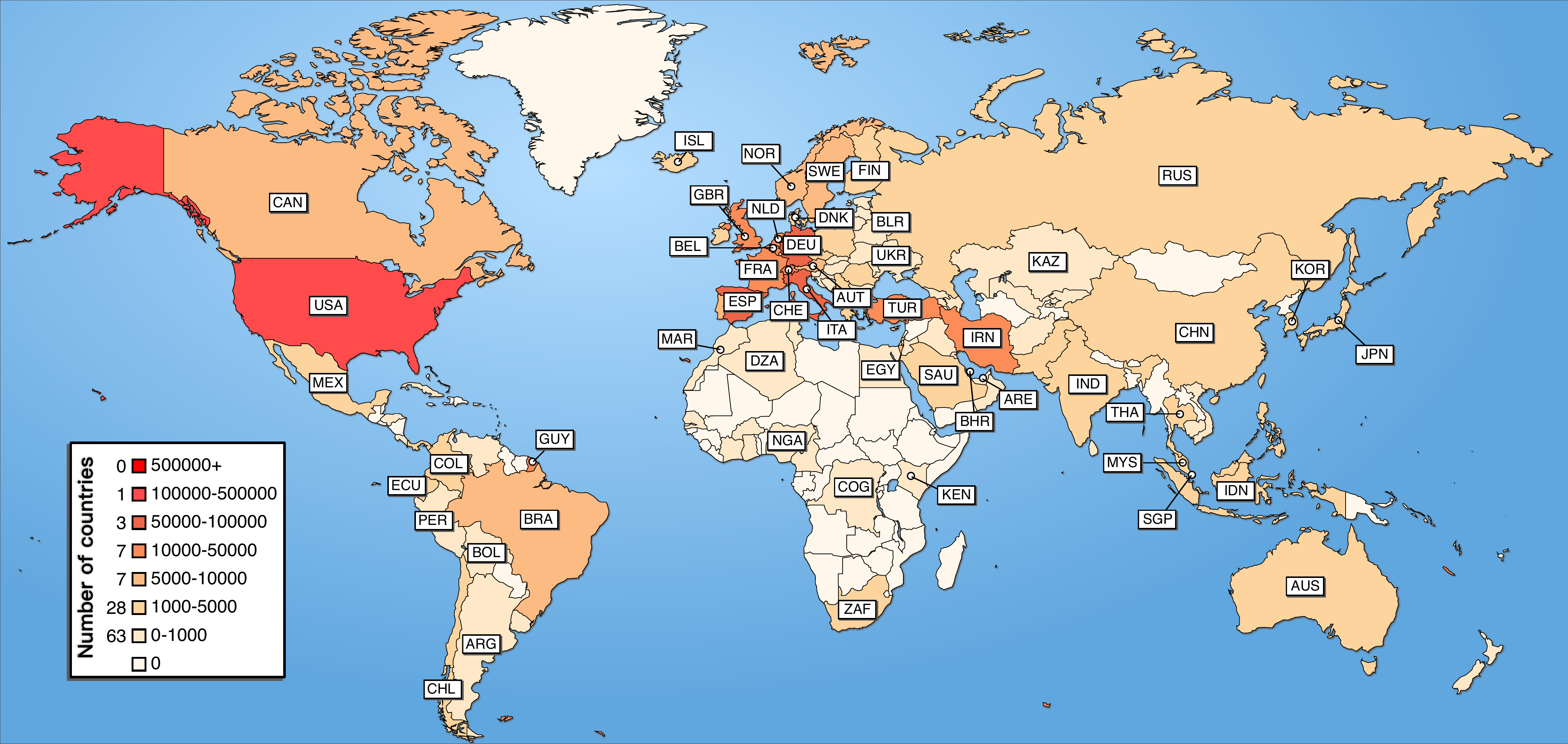}
	\caption{The number of confirmed infections $I(t)$ of the globe as of April 2, 2020.}
	\label{current}
\end{figure}

Fig. \ref{twoweekslater} shows that the trend of $I(t)$ across the globe on April 17, 2020, i.e., two weeks later. Compared with Fig. \ref{current}, there are 71 countries, 26 countries and 4 countries in the considered 109 countries where the number of confirmed infections are greater than 1,000, 10,000 and 100,000, respectively. It indicates that the global trend of $I(t)$ would become worse, especially in some European countries and USA. The prediction of $I(t)$ in USA is the maximum, which would be more than 500,000 on April 17, 2020.
\begin{figure}[H]
	\centering  
	\vspace{-0.35cm} 
	\subfigtopskip=2pt 
	\includegraphics[scale=0.35]{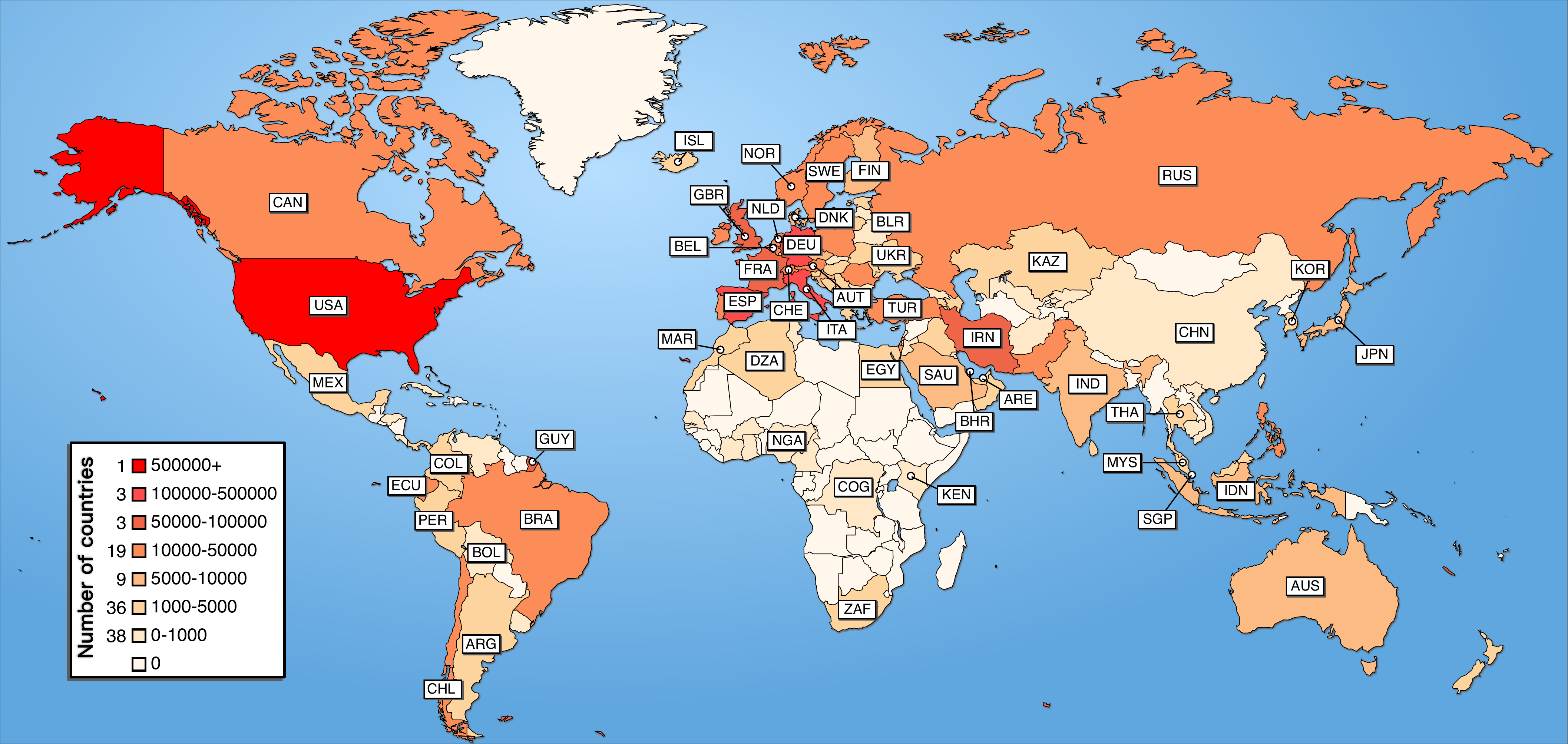}
	\caption{The prediction of the epidemic trend of  $I(t)$ around the globe on April 17, 2020 (two weeks later).}
	\label{twoweekslater}
\end{figure}

Fig. \ref{controlintensity} shows that forecasting of the number, peak and arrival time of the infections in hospitals of major countries under different control measures by changing the population size for pathogen transmission and for fixed contact rate, e.g., 1) maintaining the selected population size $N$ for pathogen transmission, 2) relaxing the control intensity and the population size $N$ has a 50\% increase, 3) strengthening the control intensity and the population size $N$ has a 50\% reduction. We find that except Korea, by decreasing of the population size, i.e., strict control measures, the peaks would have a decrease, especially in Austria, England, France, Germany, Iran, Japan, Netherlands, Spain, Switzerland and USA, and the time of peaks would arrive earlier. It is worth noting that this is different from the results in \cite{SEIR_cite3}. The reason is that the authors in \cite{SEIR_cite3} assume that the population size is fixed and strengthen the control intensity by reducing the contact rate $\alpha(t)$. Thus, it is not contradictory. In Korea, there are no effects of different population sizes. The reason is that Korea is near the end of the epidemic. There are very small number of the latent infections. The prediction of states are same for different $N$ when they are large enough. This is consistent with Propositions \ref{pro0}-\ref{pro}.

\begin{figure}[htbp]
	\centering  
	\vspace{-0.35cm} 
	\subfigtopskip=2pt 
	\subfigbottomskip=2pt 
	\subfigcapskip=-5pt 
	\subfigure[Austria]{
		\includegraphics[width=0.3\linewidth]{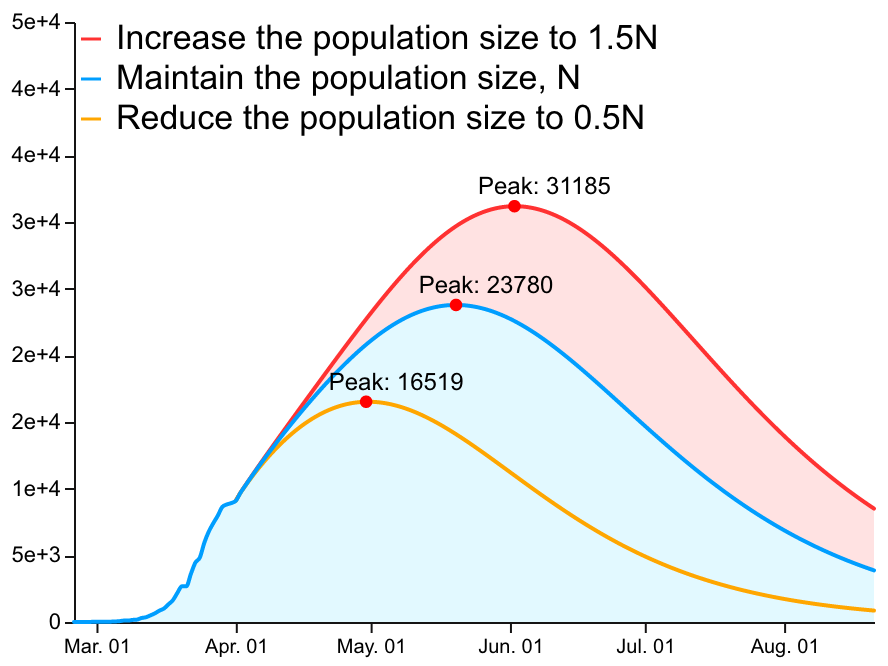}
	}
	\subfigure[England]{
		\includegraphics[width=0.3\linewidth]{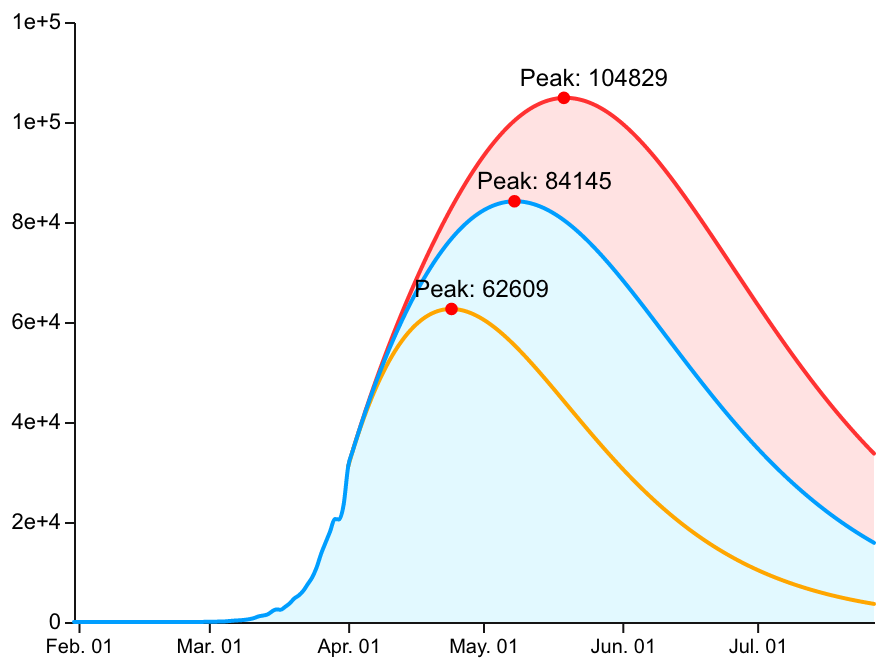}
	}%
	\subfigure[France]{
		\includegraphics[width=0.3\linewidth]{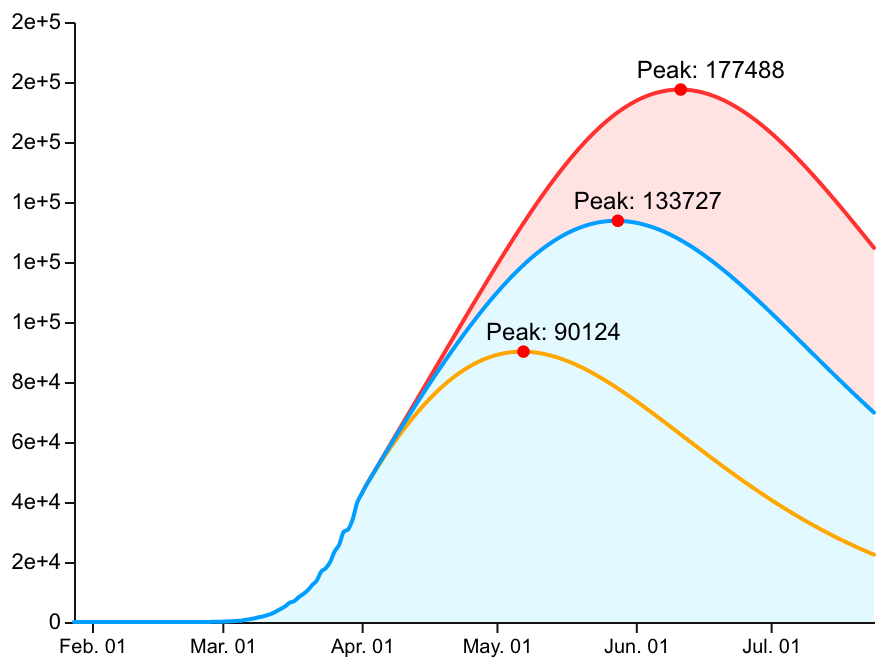}
	}

	\subfigure[Germany]{
		\includegraphics[width=0.3\linewidth]{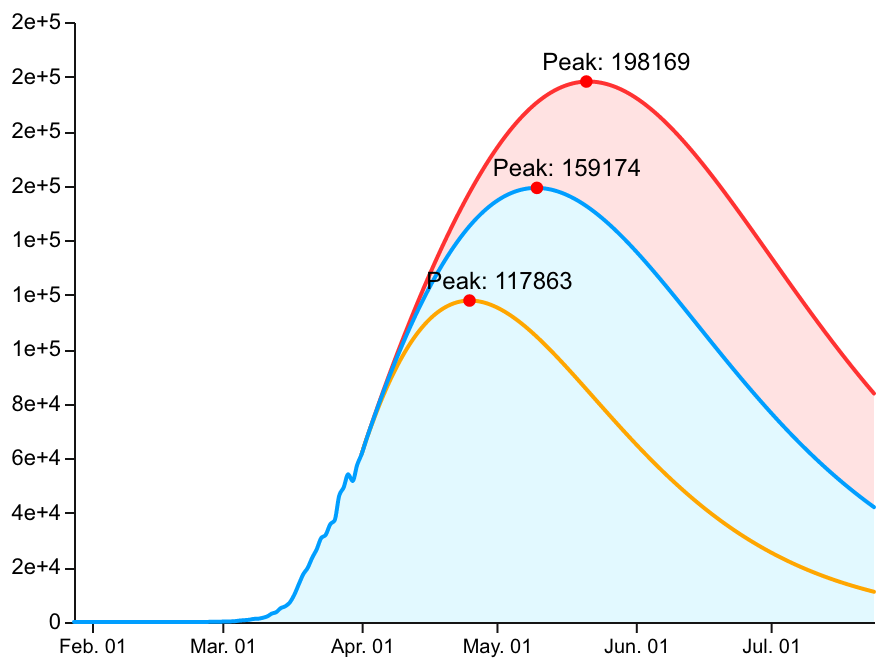}
	}%
	\subfigure[Italy]{
		\includegraphics[width=0.3\linewidth]{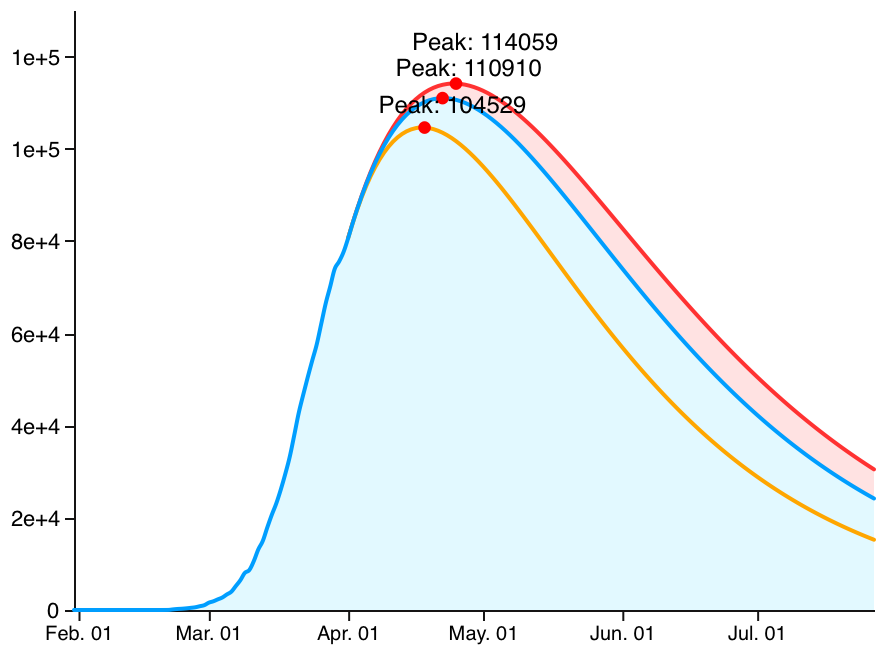}
	}%
	\subfigure[Iran]{
		\includegraphics[width=0.3\linewidth]{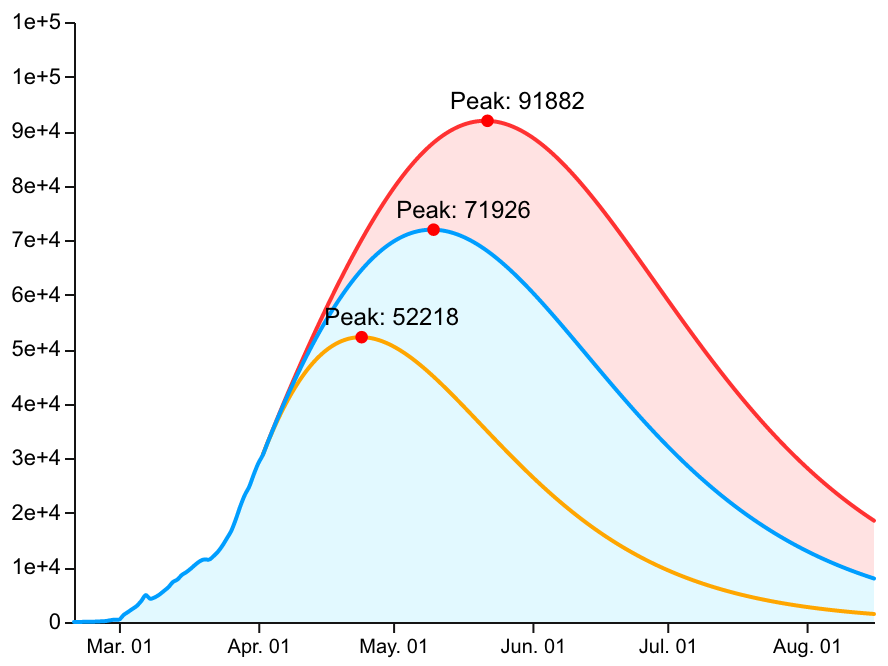}
	}

	\subfigure[Japan]{
		\includegraphics[width=0.3\linewidth]{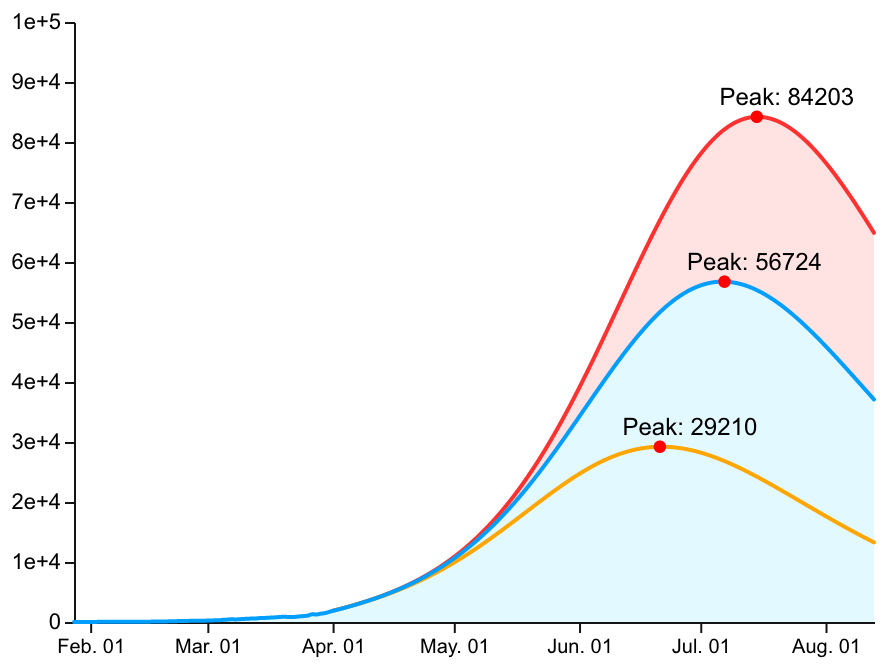}
	}%
	\subfigure[Korea]{
		\includegraphics[width=0.3\linewidth]{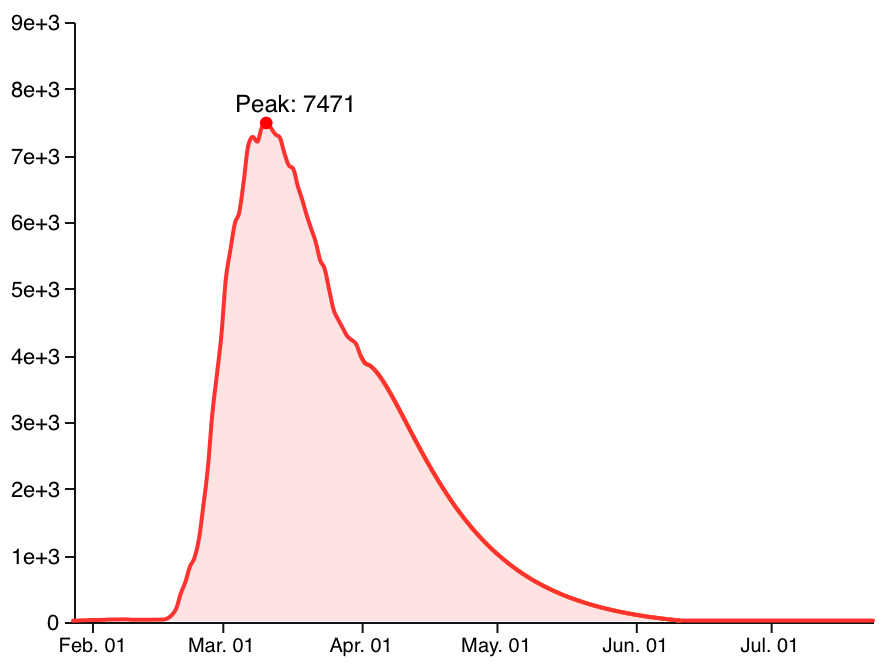}\label{C-control_Korea}
	}%
	\subfigure[Netherlands]{
		\includegraphics[width=0.3\linewidth]{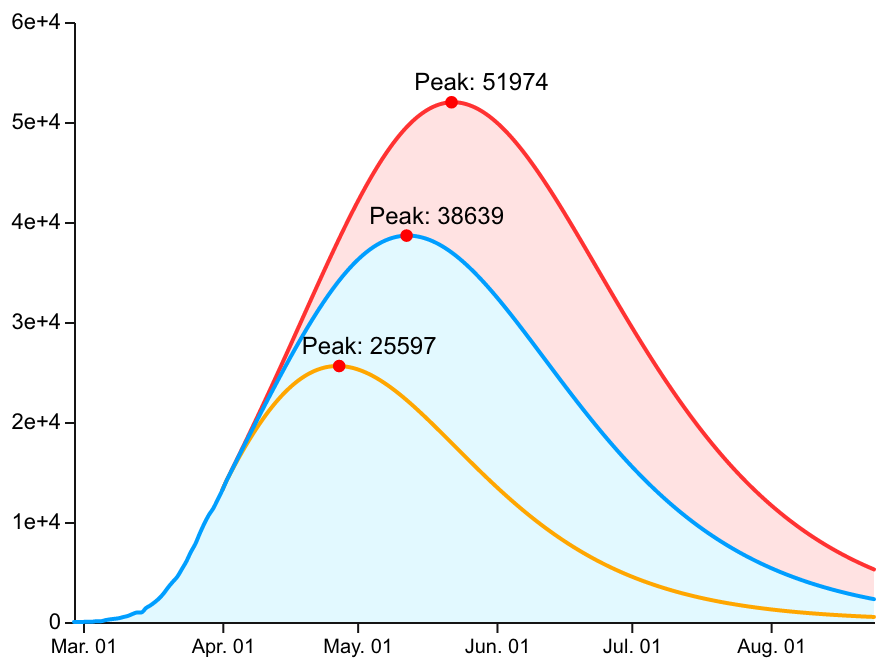}
	}%

	\subfigure[Spain]{
		\includegraphics[width=0.3\linewidth]{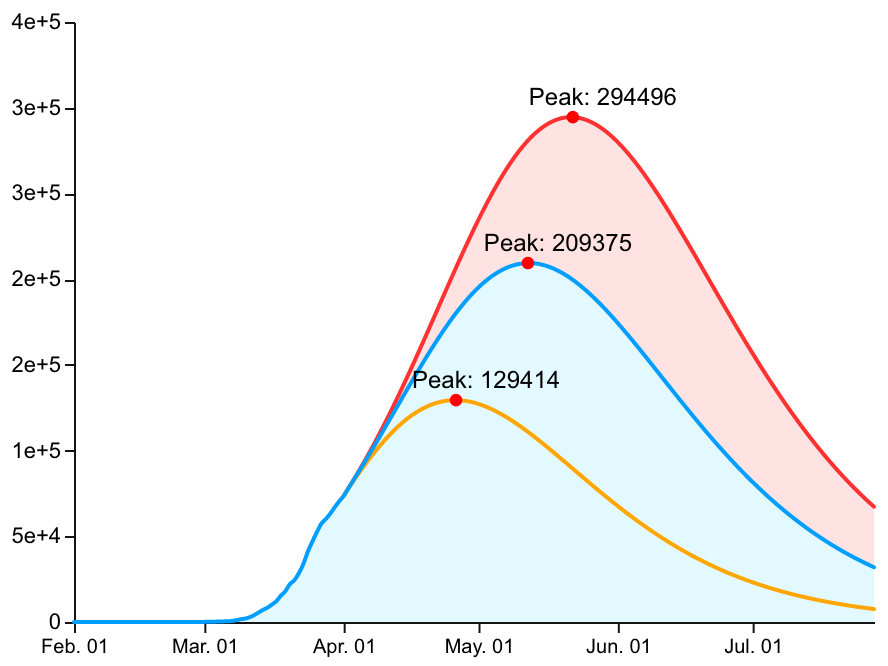}
	}%
	\subfigure[Switzerland]{
		\includegraphics[width=0.3\linewidth]{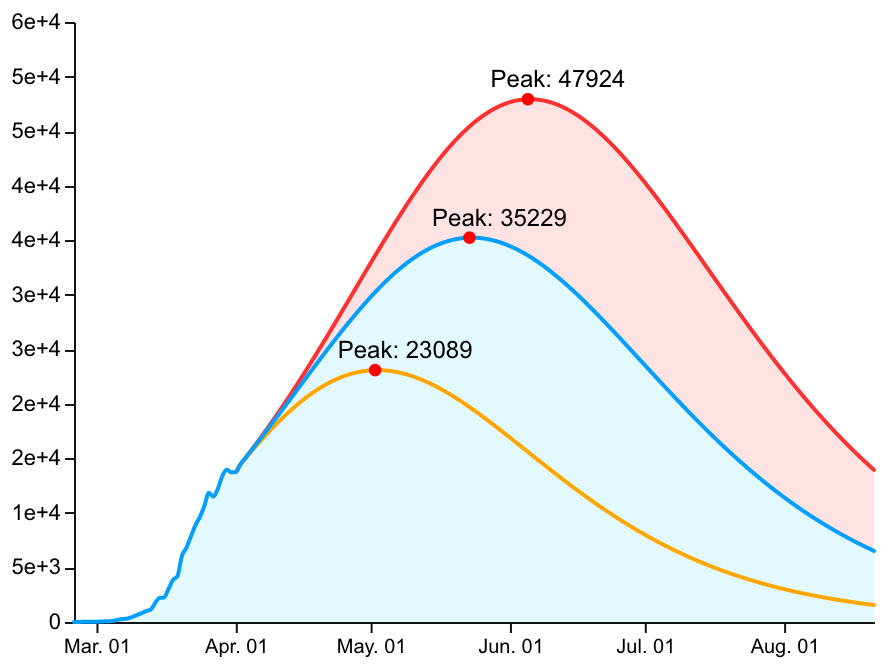}
	}%
	\subfigure[USA]{
	\includegraphics[width=0.3\linewidth]{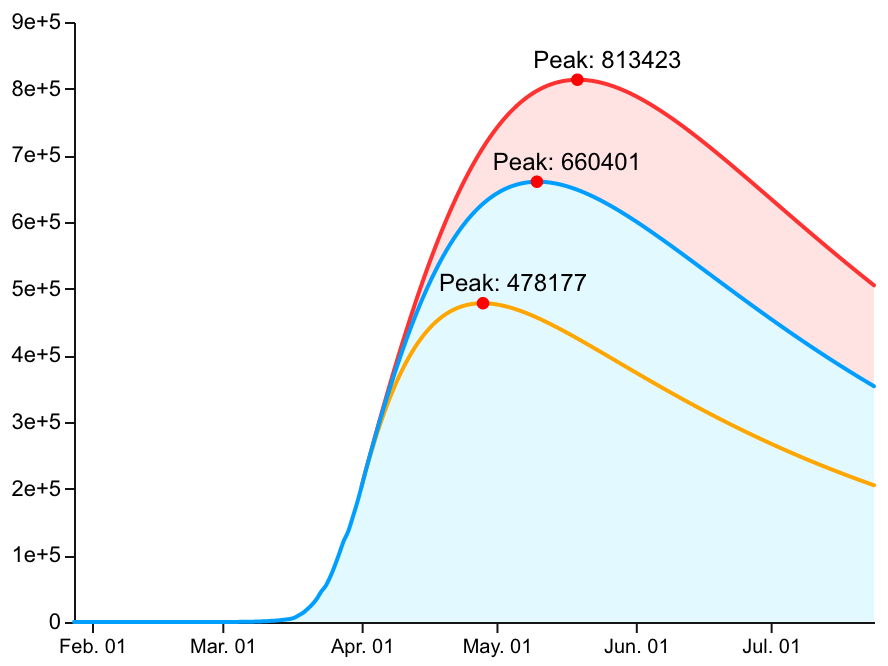}
}%
	\centering
	\caption{The impact of different control intensities implemented by changing the population size on the $I(t)$.}
	\label{controlintensity}
\end{figure}

Furthermore, Fig. \ref{com00}-\ref{com11} show the percentage of reduction of epidemic peak by decreasing the population size and the contact rate, respectively. It indicates that  Korea under different control intensities remains unchanged because it has passed the epidemic peak. However, Austria, England, France, Germany, Iran, Japan, Netherlands, Spain, Switzerland and USA under different control intensities change significantly because they are still in the outbreak period. Italy change very little because it is near the epidemic peak based on the present trajectory. 
\begin{figure}[H]
	\centering  
	\vspace{-0.35cm} 
	\subfigtopskip=2pt 
	\subfigbottomskip=2pt 
	\subfigcapskip=-5pt 
	\subfigure[Reduction of the  population size]{
		\includegraphics[scale=0.2]{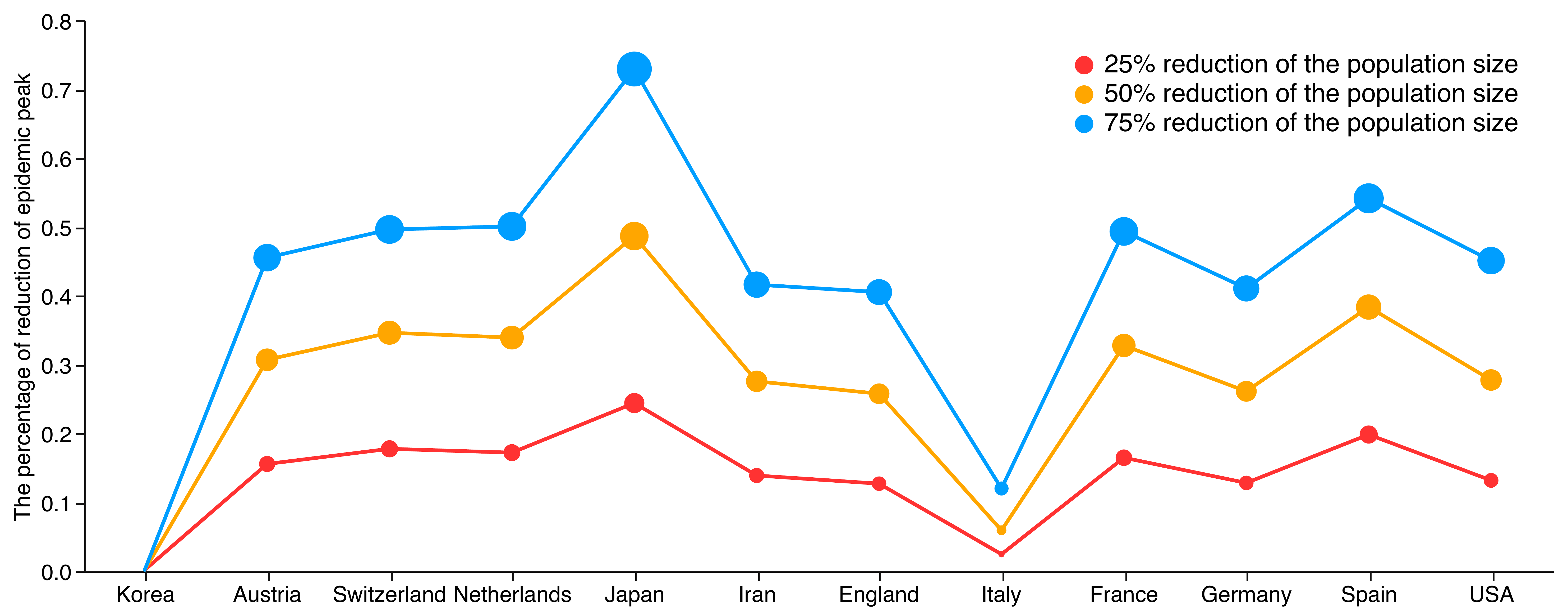}
		\label{com00}
	}
	\subfigure[Reduction of the contact rate]{
		\includegraphics[scale=0.2]{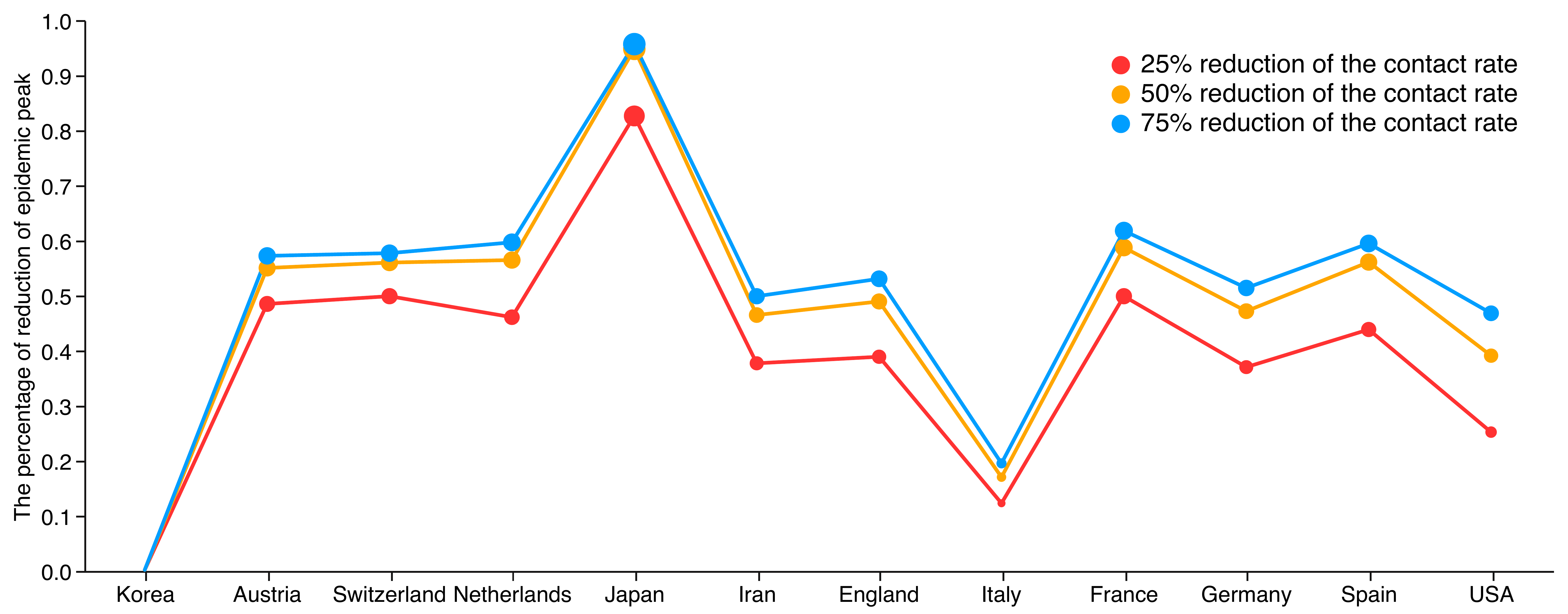}
		\label{com11}
	}%
	\centering
	\caption{The percentage of reduction of epidemic peaks of different countries under different control intensities.}
\end{figure}
\section{Discussion}
Nowcasting and forecasting of epidemics play a vital role in predicting geographic disease spread as well as case counts so as to better inform public health interventions when outbreaks occur \cite{realtimeforecast}. In this study, we developed a stochastic SEIR model under public health interventions, the nonlinear filtering and prediction  methods and selection criterion for the population size for pathogen transmission, which could be valuable to national and international agencies for public health situation perception and design of intervention strategies. Our technical contributions are made with a the backdrop of the COVID-19 pandemic that has engulfed the world. We infer that USA has the most serious situation as it has $1.90\times10^5$ (95\% CI: $1.87\times10^5$-$1.92\times10^5$) latent infections as of April 2, 2020. And in England, the latent infections account for a very large fraction as high as 58\% of the total number of infections, which would maintain outbreaks with a lag time about 1-2 weeks. And then, similar phenomenon will occur in Austria, Japan and Netherlands. The global latent infections of GEC is about $7.47\times10^5$ (95\% CI: $7.32\times10^5$-$7.62\times10^5$) as of April 2, 2020. Based on the present trajectory, the latent infections may cause further outbreaks in the case of weak public health interventions.

Therefore, it is essential to strengthen control intensity such as large-scale quarantine, strict controls on travel and extensive monitoring of suspected cases to reduce the contact rate and the population size of the susceptible. Here we regard the population size $N$ for pathogen transmission as one of the measures of the control intensity, which is conducive to assess the effectiveness of public health interventions. Moreover, it is valuable to obtain the more accurate number of the latent infections, which is helpful to design control strategies (e.g., scales of quarantine, strict controls on travel, contact isolation, hand hygiene, and use of face masks) in the future. However, on the population size $N$, we selected it by minimizing the prediction error based only on the current data, which may be prone to overfitting of the data. One possible improvement approach is to include the prediction covariance in the selection criterion. Besides, the stochastic model (\ref{stochastic SEIRC1})-(\ref{stochastic SEIRC6}) ignored the impact of the imported cases, which are considered as random noises. If the airline transportation data can be obtained, by combining the networked dynamic metapopulation model (e.g., \cite{sciencepop, metamodel, chinazzi2020effect}) and selection of the local population sizes in time, more accurate prediction methods and results for the local and global trend may be obtained. We believe that it is extremely important to predict further spread of COVID-19 globally and evaluate the effectiveness of different mitigation strategies and to try to mitigate the impact of this pandemic on the entire global population.
\section{Acknowledgement}
The authors would like to thank Xiaowei Li and Bi Yan for collecting some public data.
\section{Appendix}
\subsection{The Continuous-discrete UKF}
The main steps of the continuous-discrete UKF \cite{UKF}, \cite{cd-UKF} are summarized as follows. We define a matrix operation
\begin{align*}
\mY=g(\mX),
\end{align*} 
where $\mX\in \mathbb{R}^{n\times p}$ and $\mY\in\mathbb{R}^{n\times q}$, and the $i$-th column $\mY_i$ of the matrix $\mY$ is calculated from $g(\mX_i)$. 
\begin{itemize}
	\item \text{Model prediction step:} the state and covariance prediction $\hat{\vx}(t_{k|k-1})$, $\mP(t_{k|k-1})$  are implemented using 
	Runge-Kutta from initial conditions $\hat{\vx}(t_{k-1})$ and $\mP(t_{k-1})$:
	\begin{align}
	\label{pre1}\frac{\mathrm{d} \vx(t)}{\mathrm{d} t}=& f(\mX(t), t) w_{m}, \\
	\label{pre2}\frac{\mathrm{d} \mP(t)}{\mathrm{d} t}=& \mX(t) \mW f^{T}(\mX(t), t)+f(\mX(t), t) \mW \mX^{T}(t)+\mQ(t),
	\end{align}
	where $\vx(t)$ and $\mP(t)$ are the $n$-dimensional column vector and $n\times n$ dimensional positive semidefinite matrix, respectively. $f$ is the state transition function, the $(2n+1)\times(2n+1)$ dimensional sigma points matrix $\mX(t)$, the weight vector $w_m$ and the weight matrix $\mW$ are defined as follows:
	\begin{align}
	\label{sigma}\mX(t)=&\left(\vx(t)\quad \cdots\quad \vx(t)\right)
	\\
	&+\sqrt{n+\lambda}\left(\mathbf{0}_{n\times 1} \quad \sqrt{\mP(t)} \quad -\sqrt{\mP(t)}\right),
	\\
	w_{m}=&\left(W_{0}^{(m)} \quad \ldots  \quad W_{2n}^{(m)}\right)^{T}, \\
	\mW=& \left(\mI_{2n+1}- \left(w_{m}\quad \ldots \quad w_{m}\right)\right) \\	
	& \times \operatorname{diag}\left(W_{0}^{(c)}\quad \ldots\quad  W_{2 n}^{(c)}\right) \\
	& \times\left(\mI_{2n+1}-\left(w_{m} \quad\ldots\quad w_{m} \right)\right)^{T},\\
	W_{0}^{(m)} &=\frac{\lambda}{(n+\lambda)}, \\ W_{0}^{(c)} &=\frac{\lambda}{(n+\lambda)+\left(1-\theta^{2}+\kappa\right)}, \\ 
	W_{i}^{(m)} &=\frac{1}{2(n+\lambda)}, \quad i=1, \ldots, 2 n ,
	\\ 
	W_{i}^{(c)} &=\frac{1}{2(n+\lambda)}, \quad i=1, \ldots, 2 n ,
	\end{align}
	
	where $\sqrt{\mP(t)}$ is the matrix square root of $\mP(t)$. $\mathbf{0}_{n\times 1}, \mI_{2n+1}$ are the $n$-dimensional zero column vector, ($2n+1$)-dimensional identity matrix, respectively. $\lambda=\theta^2(n+\mu)-n$ is a parameter, and $\theta, \kappa ,\mu$  are positive constant parameters. 
		
	\item Measurement update step: Combining the new measurement $\vz_{k}$ at time $t_k$ with the state prediction $\hat{\vx}(t_{k|k-1})$ and covariance $\mP(t_{k|k-1})$ obtained in the forecast step, the state is updated via information gain matrix $\mK_{k}$,
	\begin{align}
	\label{up1}\hat{\mathbf{x}}(t_k)&=\hat{\vx}(t_{k|k-1})+\mK_{k}\left(\vz_{k}-\hat{\vz}_{k|k-1}\right), \\
	\label{up2}\mP(t_k)&=\mP(t_{k|k-1})-\mK_{k}	\operatorname{Cov}\left(\hat{\vz}_{k|k-1}\right)\mK_{k}^T,
	\end{align}
	where the gain matrix $\mK_{k}$, and the measurement covariance $\operatorname{Cov}\left(\hat{\vz}_{k|k-1}\right)$ are computed as follows:
	\begin{align}
	\label{upp1}\hat{\mX}(t_{k|k-1})=&\left(\hat{\vx}(t_{k|k-1}) \quad\cdots\quad \hat{\vx}(t_{k|k-1})\right)
	\\&+\sqrt{n+\lambda}\left(\mathbf{0}_{n\times 1} \quad \sqrt{\mP(t_{k|k-1})} \quad -\sqrt{\mP(t_{k|k-1})}\right),
	\\
	\label{measigma1}\hat{\mZ}_{k|k-1}&= \mH\hat{\mX}(t_{k|k-1}),
	\\
	\label{measigma2}\hat{\vz}_{k|k-1}&=\hat{\mZ}_{k|k-1}w_m,
	\\
	\label{meanscov}\operatorname{Cov}\left(\hat{\vz}_{k|k-1}\right)&=\hat{\mZ}_{k|k-1}\mW\hat{\mZ}_{k|k-1}^T+\mathbf{R}_{k},
	\\
	\label{cross}\operatorname{Cov}\left(\hat{\vx}(t_{k|k-1}), \hat{\vz}_{k|k-1}\right)&=\hat{\mX}(t_{k|k-1})\mW\hat{\mZ}_{k|k-1}^T,
	\\
	\label{up3}\mK_{k}&=\operatorname{Cov}\left(\hat{\vx}(t_{k|k-1}), \hat{\vz}_{k|k-1}\right)
	\operatorname{Cov}^{-1}\left(\hat{\vz}_{k|k-1}\right),
	\end{align}
\end{itemize}

More details can be seen in \cite{UKF}, \cite{cd-UKF}.

\subsection{Lemma \ref{lemma1} and its Proof }
\begin{lemma}
	\label{lemma1} Considering the deterministic SEIR model (i.e., neglecting the random noises and randomness of the deterministic parameters in (\ref{RSEIR1})-(\ref{RSEIR5})), if the states $E, I, R, D$ and model parameters have a perturbation $\mathbf{O}(\frac{1}{N})$ which is infinitesimal of the same order of $\frac{1}{N}$, and the population size $N=S+E+I+R+D$ is large enough,  then the forecast results in finite time interval derived by Runge-Kutta integration are with perturbation $\mathbf{O}(\frac{1}{N})$ as well.
\end{lemma}
\begin{proof}
	Since $N=S+E+I+R+D$, the deterministic SEIR model  can be simplified as follows:
	\begin{align}
	\label{deter1}\frac{\mathrm{d} E(t)}{\mathrm{d} t}&=\alpha \cdot \frac{(N-E(t)-I(t)-R(t)-D(t)) \cdot I(t)}{N}-\beta\cdot E(t),\\
	\frac{\mathrm{d} I(t)}{\mathrm{d} t}&=\beta\cdot E(t) -(\gamma^1+\gamma^2) \cdot I(t),\\
	\frac{\mathrm{d} R(t)}{\mathrm{d} t}&=\gamma^1 \cdot I(t),\\
	\label{deter5}\frac{\mathrm{d} D(t)}{\mathrm{d} t}&=\gamma^2 \cdot I(t),
	\end{align}
	where the parameters $\alpha$, $\beta$, $\gamma^1$, $\gamma^2$ are constants, and $S=N-E-I-R-D$ holds for all time. Denote the state with a perturbation as 
	\begin{align*}
	\tilde{\vx}(t_k)&=\vx(t_k)+\mathbf{O}(\frac{1}{N}),
	\end{align*}
	where $\vx(t_k)$ represents the true state $\left(E(t_k), I(t_k), R(t_k), D(t_k), \alpha, \beta, \gamma^1, \gamma^2\right)$ at time $t_k$. We use the convention $f_d(\vx(t),t)$ as the drift function in differential equations (\ref{deter1})-(\ref{deter5}). Herein, we adopt the fourth order Runge-Kutta method to derive the state  $\tilde{\vx}(t_m)$ at time $t_m$ by using initial condition $\tilde{\vx}(t_k)$.
	\begin{align}
	\label{RK4}\tilde{\vx}(t_m)=&\tilde{\vx}(t_k)+\frac{h}{6}(\tilde{k}_1+2\tilde{k}_2+2\tilde{k}_3+\tilde{k}_4),
	\\
	\tilde{k}_1&=f_d(\tilde{\vx}(t_k),t_k),
	\\
	\label{k2}\tilde{k}_2&=f_d(\tilde{\vx}(t_k)+\frac{h}{2}\tilde{k}_1,t_k+\frac{h}{2}),
	\\
	\tilde{k}_3&=f_d(\tilde{\vx}(t_k)+\frac{h}{2}\tilde{k}_2,t_k+\frac{h}{2}),
	\\
	\label{k4}\tilde{k}_4&=f_d(\tilde{\vx}(t_k)+h\tilde{k}_3,t_k+h),
	\end{align}
	where $h$ is the time interval $t_m-t_k$. By ignoring the higher order terms than $\mathbf{O}(\frac{1}{N})$ generated in the nonlinear transformation of $f_d$, we have
	\begin{align}
	\begin{split}
	\label{divide}\tilde{k}_1&=f_d(\tilde{\vx}(t_k),t_k)\\
	&=f_1(\vx(t_k),t_k)+f_N(\vx(t_k),t_k),
	\end{split}
	\end{align}
	where $f_d(\tilde{\vx}(t_k),t_k)$, $f_N(\vx(t_k),t_k)$ and $f_1(\vx(t_k),t_k)$ are
	\begin{align*}
	f_d(\tilde{\vx}(t_k),t_k)&=\left(      
	\begin{array}{c}
	-(\beta+\mathbf{O}(\frac{1}{N})(E(t_k)+\mathbf{O}(\frac{1}{N})+\frac{(\alpha+\mathbf{O}(\frac{1}{N}))(I(t_k)+\mathbf{O}(\frac{1}{N}))}{N}\cdot
	\\
	(N-E(t_k)-I(t_k)-R(t_k)-D(t_k)-\mathbf{O}(\frac{1}{N}))
	\\
	(\beta+\mathbf{O}(\frac{1}{N})(E(t_k)+\mathbf{O}(\frac{1}{N})-(\gamma^1+\gamma^2+\mathbf{O}(\frac{1}{N}))(I(t_k)+\mathbf{O}(\frac{1}{N})
	\\
	(\gamma^1+\mathbf{O}(\frac{1}{N})(I(t_k)+\mathbf{O}(\frac{1}{N}) 
	\\
	(\gamma^2+\mathbf{O}(\frac{1}{N})(I(t_k)+\mathbf{O}(\frac{1}{N}) 
	\end{array}\right),
	\\ 
	f_N(\vx(t_k),t_k)&=\left(      
	\begin{array}{c}  
	\alpha \mathbf{O}(\frac{1}{N})+ I(t_k)\mathbf{O}(\frac{1}{N})-\beta\mathbf{O}(\frac{1}{N})-E(t_k)\mathbf{O}(\frac{1}{N})-
	\\
	\alpha(E(t)+I(t)+R(t)+D(t)I(t)\cdot \mathbf{O}(\frac{1}{N})
	\\
	\beta\mathbf{O}(\frac{1}{N})+E(t_k)\mathbf{O}(\frac{1}{N})-(\gamma^1+\gamma^2)\mathbf{O}(\frac{1}{N})-I(t_k)\mathbf{O}(\frac{1}{N})
	\\
	\gamma^1\mathbf{O}(\frac{1}{N})+I(t_k)\mathbf{O}(\frac{1}{N})
	\\
	\gamma^2\mathbf{O}(\frac{1}{N})+I(t_k)\mathbf{O}(\frac{1}{N})
	\end{array}\right),
	\\
	f_1(\vx(t_k),t_k)&=\left(      
	\begin{array}{c}  
	\alpha I(t_k)-\beta E(t_k)
	\\
	\beta E(t_k)-(\gamma^1+\gamma^2) I(t_k)
	\\
	\gamma^1 I(t_k)
	\\
	\gamma^2 I(t_k)
	\end{array}\right).
	\end{align*}
	Similarly, the $k_1$ calculated by the initial value $\vx(t_k)$ without the permutation $\mathbf{O}(\frac{1}{N})$ is 
	\begin{align*}
	\begin{split}
	{k}_1&=f_d(\vx(t_k),t_k)\\
	&=\overline{f_N}(\vx(t_k),t_k)+f_1(\vx(t_k),t_k),
	\end{split}
	\end{align*}
	where $\overline{f_N}(\vx(t_k),t_k)$ is
	\begin{align*}
	\overline{f_N}(\vx(t_k),t_k)=\left(      
	\begin{array}{c}  
	-\alpha(E(t)+I(t)+R(t)+D(t)I(t)\cdot \mathbf{O}(\frac{1}{N})
	\\
	0
	\\
	0
	\\
	0
	\end{array}\right),
	\end{align*}
	Comparing $\tilde{k}_1$ and $k_1$, we conclude that $\tilde{k}_1=k_1+\mathbf{O}(\frac{1}{N})$. Similarly from the calculation of (\ref{k2})-(\ref{k4}), it can be concluded that $\tilde{k}_i=k_i+\mathbf{O}(\frac{1}{N}), i=2, 3, 4$ as well. By (\ref{RK4}), we have
	\begin{align*}
	\tilde{\vx}(t_m)=\vx(t_m)+\mathbf{O}(\frac{1}{N}).
	\end{align*}
\end{proof}

\subsection{Proposition \ref{pro0} and its Proof}

\begin{proposition}
	\label{pro0} If the current state estimation and covariance at time $t_k$  have a perturbation  $\mathbf{O}(\frac{1}{N})$ which has a compatible dimension and the population size $N=S+E+I+R+D$ is large enough, then the state prediction $\hat{\vx}(t)$ and covariance $\mP(t)$ in finite time derived by the model prediction step of the continuous-discrete UKF via Runge-Kutta integration are with perturbations $\mathbf{O}(\frac{1}{N})$ as well.
\end{proposition}
\begin{proof}
	Considering the model prediction step, the Runge-Kutta integration is implemented for the differential equations (\ref{pre1})-(\ref{pre2}) to derive the state prediction and covariance. Without loss of generality, we assume that the state and covariance at time $t_k$ are with a perturbation $\mathbf{O}(\frac{1}{N})$,
	\begin{align*}
	\tilde{\vx}(t_k)&=\hat{\vx}(t_k)+\mathbf{O}(\frac{1}{N}),
	\\
	\tilde{\mP}(t_k)&=\mP(t_k)+\mathbf{O}(\frac{1}{N}),
	\end{align*}
	where $\hat{\vx}(t_k)$, $\mP(t_k)$ are the state estimation and covariance at time $t_k$. By the continuous-discrete UKF, the set of sigma points is chosen via (\ref{sigma}), 
	\begin{align}
	\hat{\mX}(t_k)&=\left(\hat{\vx}(t_k)\quad \cdots\quad \hat{\vx}(t_k)\right)+\sqrt{n+\lambda}\left(\mathbf{0} \quad \sqrt{\mP(t_k)} \quad -\sqrt{\mP(t_k)}\right),\\
	\label{proof1}\tilde{\mX}(t_{k})&=\left(\tilde{\vx}(t_k)\quad \cdots\quad \tilde{\vx}(t_k)\right)+\sqrt{n+\lambda}\left(\mathbf{0} \quad \sqrt{\tilde{\mP}(t_k)} \quad -\sqrt{\tilde{\mP}(t_k)}\right),
	\end{align}
	where $\hat{\mX}(t_{k})$ and $\tilde{\mX}(t_{k})$ are derived by the state and covariance without perturbations and with perturbations, respectively.

	As is shown in \cite{matrixdecom}, in light of  the first order perturbation bound of Cholesky decomposition, we have:
	\begin{align*}
	\frac{\|\sqrt{\tilde{\mP}(t_k)}-\sqrt{\mP(t_k)}\|_{F}}{\|\sqrt{\mP(t_k)}\|} \leq \frac{1}{\sqrt{2}} \operatorname{cond}(\mP(t_k)) \frac{\|\tilde{\mP}(t_k)-\mP(t_k)\|_{F}}{\|\mP(t_k)\|},
	\end{align*}
	where $\|\cdot\|_{F}$ is the Frobenius norm, and $\operatorname{cond}(\cdot)$ represents the condition number. Since $\mP(t_k)$ is a constant matrix and $ \tilde{\mP}(t_k)-\mP(t_k)=\mathbf{O}(\frac{1}{N})$, we have
	\begin{align*}
	\sqrt{\tilde{\mP}(t_k)}=\sqrt{\mP(t_k)}+\mathbf{O}(\frac{1}{N}).
	\end{align*}
	Rewriting (\ref{proof1}), we obtain that 
	\begin{align}
	\label{sigmapoint}\tilde{\mX}(t_{k})&=\hat{\mX}(t_k)+\mathbf{O}(\frac{1}{N}),
	\end{align}
	namely, the sigma points of $\tilde{\vx}(t_k)$ are also with a perturbation $\mathbf{O}(\frac{1}{N})$ compared with the sigma points of $\hat{\vx}(t_k)$.
	
	Thus, we adopt the fourth-order Runge-Kutta method to calculate the state and covariance prediction (\ref{pre1})-(\ref{pre2}). By Lemma \ref{lemma1}, we have that the state and covariance prediction $\tilde{\vx}(t_{k+1|k})$, $\tilde{\mP}(t_{k+1|k})$ based on $\tilde{\vx}(t_{k}),\tilde{\mP}(t_{k})$ are still with a perturbation $\mathbf{O}(\frac{1}{N})$ comparing with the state and covariance $\hat{\vx}(t_{k+1|k})$, $\mP(t_{k+1|k})$ predicted by $\hat{\vx}(t_{k})$ and $\mP(t_{k})$.	
	
\end{proof}
\subsection{Proposition \ref{pro} and its Proof}
\begin{proposition}
	\label{pro}If the state estimation and covariance at time $t_k$ are with perturbations $\mathbf{O}(\frac{1}{N})$ which  has a compatible dimension and the population size $N=S+E+I+R+D$ is large enough, then the state update of $\hat{\vx}(t)$ and $\mP(t)$ in finite time through continuous-discrete UKF via Runge-Kutta integration are still with perturbations $\mathbf{O}(\frac{1}{N})$.
\end{proposition}
\begin{proof}

	Denote the state estimation and covariance with  perturbations as follows:
	\begin{align*}
	\tilde{\vx}(t_k)=&\hat{\vx}(t_k)+\mathbf{O}(\frac{1}{N}),
	\\
	\tilde{\mP}(t_k)=&\mP(t_k)+\mathbf{O}(\frac{1}{N}),
	\end{align*}
	By Proposition \ref{pro0}, we know that the state and covariance prediction of $\tilde{\vx}(t_k), \tilde{\mP}(t_k)$ are still with a perturbation, i.e.,
	\begin{align*}
	\tilde{\vx}(t_{k+1|k})=& 
	\hat{\vx}(t_{k+1|k})+\mathbf{O}(\frac{1}{N}),
	\\
	\tilde{\mP}(t_{k+1|k})=&
	\mP(t_{k+1|k})+\mathbf{O}(\frac{1}{N}),
	\end{align*}
	then the state prediction and covariance $\hat{\vx}(t_{k+1|k}), \mP(t_{k+1|k})$ without perturbation are updated by measurement $\vz_{k+1}$ through (\ref{up1})-(\ref{up2}), i.e.,
	\begin{align*}
	\hat{\mathbf{x}}(t_{k+1})&=\hat{\vx}(t_{k+1|k})+\mK_{k+1}\left(\vz_{k+1}-\hat{\vz}_{k+1|k}\right), \\
	\mP(t_{k+1})&=\mP(t_{k+1|k})-\mK_{k+1}	\operatorname{Cov}\left(\hat{\vz}_{k+1|k}\right)\mK_{k+1}^T,
	\end{align*}
	where $\hat{\vz}_{k+1|k}, \operatorname{Cov}\left(\hat{\vz}_{k+1|k}\right), \mK_{k+1}$ are calculated by (\ref{upp1})-(\ref{up3}).
	
	Besides, for the state prediction and covariance $\tilde{\vx}(t_{k+1|k}), \tilde{\mP}(t_{k+1|k})$ with a perturbation, similarly as (\ref{sigmapoint}), the sigma points of $\tilde{\vx}(t_{k+1|k})$ are with perturbations, i.e.,
	\begin{align*}
	\tilde{\mX}(t_{k+1|k})=\hat{\mX}(t_{k+1|k})+\mathbf{O}(\frac{1}{N}),
	\end{align*}
	by the observation model (\ref{cd2}) and (\ref{measigma1})-(\ref{measigma2}), we have
	\begin{align*}
	\tilde{\mZ}_{k+1|k}&=\mH\tilde{\mX}(t_{k+1|k})\\
	&=\hat{\mZ}_{k+1|k}+\mathbf{O}(\frac{1}{N})\\
	\tilde{\vz}_{k+1|k}&=\mH\hat{\mX}(t_{k+1|k})w_{m}+\mathbf{O}(\frac{1}{N}),\\
	&=\hat{\vz}_{k+1|k}+\mathbf{O}(\frac{1}{N})
	\end{align*}
	namely, the observations of $\tilde{\mX}(t_{k+1|k})$ are still with a perturbation $\mathbf{O}(\frac{1}{N})$ comparing to that of $\hat{\mX}(t_{k+1|k})$. The innovation covariance and cross covariance are calculated by (\ref{meanscov}) and (\ref{cross}),
	\begin{align*}
	\operatorname{Cov}\left(\tilde{\vz}_{k+1|k}\right)&=\tilde{\mZ}_{k+1|k}\mW\tilde{\mZ}_{k+1|k}^T+\mR_{k+1}
	\\
	&=\operatorname{Cov}\left(\hat{\vz}_{k+1|k}\right)+\mathbf{O}(\frac{1}{N}),\\
	\operatorname{Cov}\left(\tilde{\vx}(t_{k+1|k}),\tilde{\vz}_{k+1|k}\right)&=\tilde{\mX}(t_{k+1|k})\mW\tilde{\mZ}_{k+1|k}^T
	\\
	&=\operatorname{Cov}\left(\hat{\vx}(t_{k+1|k}),\hat{\vz}_{k+1|k}\right)+\mathbf{O}(\frac{1}{N}).
	\end{align*}
	According to \cite{matrixinv}, the inverse of matrix with perturbations has following property:
	\begin{align*}
	\left\|B^{-1}-A^{-1}\right\|_F \leqslant \mu \left\|A^{-1}\right\|_2\left\|B^{-1}\right\|_{2} \left\|B-A\right\|_{F},
	\end{align*}
	where $\mu$ is a constant, $\left\|\cdot\right\|_2$ and $\left\|\cdot\right\|_F$ represent the Euclidean norm and Frobenius norm, respectively. If we denote that  $B=\operatorname{Cov}\left(\tilde{\vz}_{k+1|k}\right)$, $A=\operatorname{Cov}\left(\hat{\vz}_{k+1|k}\right)$, then we have 
	\begin{align*}
	\left\|\operatorname{Cov}\left(\tilde{\vz}_{k+1|k}\right)^{-1}-\operatorname{Cov}\left(\hat{\vz}_{k+1|k}\right)^{-1}\right\|_F\leqslant
	\mu\cdot\left\|\operatorname{Cov}\left(\hat{\vz}_{k+1|k}\right)^{-1}\right\|&_2\cdot\left\|\operatorname{Cov}\left(\tilde{\vz}_{k+1|k}\right)^{-1}\right\|_2\cdot\left\|\mathbf{O}(\frac{1}{N})\right\|_{F},
	\end{align*}
	Thus,
	\begin{align*}
	\operatorname{Cov}^{-1}\left(\tilde{\vz}_{k+1|k}\right)=\operatorname{Cov}\left(\hat{\vz}_{k+1|k}\right)^{-1}+\mathbf{O}(\frac{1}{N}),
	\end{align*}
	and the gain matrix $\tilde{\mathbf{K}}_{k+1}$ is
	\begin{align*}
	\tilde{\mathbf{K}}_{k+1}&=\operatorname{Cov}\left(\tilde{\vx}(t_{k+1|k}),\tilde{\vz}_{k+1|k}\right)\operatorname{Cov}^{-1}\left(\tilde{\vz}_{k+1|k}\right)
	\\
	&=\mathbf{K}_{k+1}+\mathbf{O}(\frac{1}{N}).
	\end{align*}
	Moreover, the state prediction and covariance are updated by measurement $\vz_{k+1}$ as follows:
	\begin{align*}
	\tilde{\mathbf{x}}(t_{k+1})&=\tilde{\vx}(t_{k+1|k})+\tilde{\mathbf{K}}_{k+1}\left(\vz_{k+1}-\tilde{\vz}_{k+1|k}\right) 
	\\
	&=\hat{\vx}(t_{k+1})+\mathbf{O}(\frac{1}{N}),
	\\
	\tilde{\mP}(t_{k+1})&=\tilde{\mP}(t_{k+1|k})-\tilde{\mathbf{K}}_{k+1}	\operatorname{Cov}^{-1}\left(\tilde{\vz}_{k+1|k}\right)\tilde{\mathbf{K}}_{k+1}^T
	\\
	&=\hat{\mP}(t_{k+1})+\mathbf{O}(\frac{1}{N}).
	\end{align*}
	
	Therefore, the state estimation and covariance at time $t_{k+1}$ are with perturbations $\mathbf{O}(\frac{1}{N})$ as well.
\end{proof}

\end{document}